\documentclass[submission,copyright,creativecommons]{eptcs}
\usepackage{etex}
\providecommand{\event}{ICE 2016} 

\usepackage{bussproofs,enumerate}

\usepackage[utf8]{inputenc}
\usepackage{amsthm}
\usepackage{amsmath}
\usepackage{mik-makro}
\usepackage{graphicx}
\usepackage{xspace}
\usepackage[all]{xy}
\usepackage{float}
\usepackage{todonotes}
\usepackage[normalem]{ulem}
\usepackage{tikz}
\usepackage{tikz-qtree}
\usepackage{enumitem}
\usepackage[font={footnotesize,it}]{caption}
\usepackage{subcaption}

\newtheorem{theorem}{Theorem}
\newtheorem{definition}[theorem]{Definition}
\newtheorem{example}[theorem]{Example}

\newtheorem{corollary}[theorem]{Corollary}
\newtheorem{proposition}[theorem]{Proposition}

\newtheorem{claim}[theorem]{Claim}
\newtheorem*{subclaim*}{Subclaim}
\newtheorem*{remark*}{Remark}

\usetikzlibrary{trees,automata,positioning}
\usetikzlibrary{arrows,decorations.pathmorphing,backgrounds,positioning,fit,petri}
\usetikzlibrary{automata}
\usetikzlibrary{graphs}
\usetikzlibrary{shadows}

\newcommand{\Qall}{Q_{\mathrm{all}}}
\newcommand{\Qzero}{Q_{\mathrm{zero}}}

\newcommand{\trees}[1]{\mathsf{trees}(#1)}

\newcommand{\ignore}[1]{}

\newcommand{\zero}{\texttt {subzero}}

\renewcommand{\Aa}{{\mathcal A}}
\newcommand{\N}{{\mathbb N}}

\newcommand{\qall}{Q_\texttt{all}}
\newcommand{\qzero}{Q_\texttt{zero}}

\newcommand{\maxinf}{\texttt{maxinf}}

\newcommand{\synseq}[3]{#1 \stackrel {\le #2} \longrightarrow #3}
\newcommand{\synsneq}[3]{#1 \stackrel {< #2} \longrightarrow #3}

\newcommand{\profile}[3]{#1 \stackrel {\le #2} \longrightarrow #3}

\title{On the Regular Emptiness Problem of Subzero Automata\footnote{Henryk Michalewski was supported by Poland’s National Science Centre grant no.~2012/07/D/ST6/02443 and Matteo Mio was supported by``Projet \'{E}mergent PMSO'' of the \'{E}cole Normale Sup\'{e}rieure de Lyon and Poland’s National Science Centre grant no.~2014-13/B/ST6/03595.}}
\author{
Henryk Michalewski
\institute{University of Warsaw, Poland}
\and
Matteo Mio
\institute{CNRS and ENS-Lyon, France}
\and
Mikołaj Bojańczyk
\institute{University of Warsaw, Poland}
}
\def\titlerunning{On the Regular Emptiness Problem of Subzero Automata}
\def\authorrunning{Michalewski, Mio\ \& Bojańczyk}

\begin{document}

\maketitle
\begin{abstract} Subzero automata is a class of tree automata whose acceptance condition can express probabilistic constraints. 
Our main result is that the problem of determining if a subzero automaton accepts some regular tree is decidable.
\end{abstract}	

\section{Introduction}

In the  fundamental paper \cite{Rabin69} Rabin proved that the monadic second order logic (MSO) of the full binary tree is decidable using the \emph{automata method}. This proof technique can be roughly described as follows: (1) an appropriate notion of \emph{tree automaton} is defined; (2) every formula $\phi$ of the MSO logic is effectively 
associated to an automaton $\Aa_\phi$ such that that $\phi$ is true if and only if $\Aa_\phi$ accepts a non-empty language; 
(3) the \emph{emptiness problem}, that is deciding if $\Aa_\phi$ accepts a non-empty language, is proved to be decidable. The latter point is typically established using combinatorial reasoning about the graph structure of $\Aa_\phi$. 

Recently, Michalewski and Mio have investigated in \cite{DBLP:conf/lfcs/MichalewskiM16} an extension of the MSO logic of the full binary tree, called $\textnormal{MSO}+\forall^{=1}$, capable of expressing probabilistic properties. While the full logic $\textnormal{MSO}+\forall^{=1}$ is undecidable (see Section 5 of \cite{DBLP:conf/lfcs/MichalewskiM16}), the decidability of an interesting fragment called $\textnormal{MSO}+\forall^{=1}_\pi$ (see Section 6 of \cite{DBLP:conf/lfcs/MichalewskiM16}), capable of expressing many probabilistic properties useful in program verification such as those definable by the logic pCTL and its variants,  is an open problem.

Boja{\'n}czyk proposed in \cite{bojanzero2016} to use the automata method to prove the decidability of the weak fragment of $\textnormal{MSO}+\forall^{=1}_\pi$, where second-order quantification is restricted to finite sets, which is still sufficiently expressive to express most useful probabilistic properties in program verification.
Namely, Boja{\'n}czyk has: (1) introduced a special class of \emph{zero automata} 
and (2) proved that for every formula $\phi$ of  $\textnormal{weakMSO}+\forall^{=1}_\pi$ one can effectively associate 
a zero automaton $\Aa_\phi$ such that $\phi$ is true if and only if $\Aa_\phi$ accepts a non-empty language. 

Hence what is still missing from Boja{\'n}czyk's approach is a proof of decidability of the emptiness problem of zero automata. In this paper we consider a simplified version of zero automata introduced by Boja{\'n}czyk\footnote{This simplification  makes the presentation smoother. At the same time, the regular emptiness problem for subzero automata seems to be equally difficult as the regular emptiness problem for zero automata. Admittedly, in order to characterize the full strength of $\textnormal{weakMSO}+\forall^{=1}_\pi$ one needs a class of automata more refined than the one analyzed in this paper. We accept this drawback in order to keep the presentation as simple as possible.}.  We call the simplified class  \emph{subzero automata} and prove the following result:

\begin{theorem}\label{main_result_1}
Given a subzero automaton $\mathcal{A}$, it is decidable if there exists a regular tree (i.e., representable as a finite directed graph) which is accepted by $\mathcal{A}$.
\end{theorem}

However we also observe that:

\begin{proposition}\label{proposition_irregularity}
There exists a subzero automaton $\mathcal{A}$ such that the language accepted by $\Aa$ is not empty but does not contain a regular tree.
\end{proposition}

Recall that Rabin's regularity theorem  \cite{Rabin69}  asserts that a Rabin's tree automaton accepts some tree if and only if it accepts some regular tree. Hence Proposition \ref{proposition_irregularity} states that subzero automata do not enjoy this nice property. In turn, this fact implies that our result (Theorem \ref{main_result_1}) does not yet solve the emptiness problem of subzero automata. 

Nevertheless the proof method we use to prove Theorem \ref{main_result_1} is quite interesting as it not based on two-player (stochastic) games, graph algorithms or similar techniques based on combinatorics on graphs, but on the design of a \emph{deduction system} capable of deriving assertions such as ``from state $s$ the automaton $\mathcal{A}$ has an accepting run''. To the best of our knowledge this approach has not yet been used in the literature and we believe that it can be applied to the emptiness problems of other classes of automata.
Indeed, our preliminary investigations indicate that also the full result, the decidability of the (not just regular) emptiness problem of zero automata, might be provable by the design of a (significantly more complicated) proof system.  The purpose of this paper is to illustrate the main ideas behind the use of deductive proof systems 
 in the simpler context of the proof of Theorem \ref{main_result_1}.


\paragraph{Related work.}
In \cite{DBLP:journals/tocl/CarayolHS14},  Carayol, Haddad and  Serre have introduced the new automata model of \emph{qualitative tree automata}, a variant of Rabin's tree automata with a probabilistic acceptance condition. 
Main results from  \cite{DBLP:journals/tocl/CarayolHS14} about qualitative automata include: (1) the class of languages definable by qualitative automata is incomparable with the class of regular languages, i.e., there are some languages definable by qualitative automata that are not definable by Rabin's automata, and vice versa, (2) qualitative automata enjoy the Rabin's regularity theorem, that is they accept some tree if and only if they accept some regular tree, (3)  the emptiness of qualitative automata is decidable.

The work of \cite{DBLP:journals/tocl/CarayolHS14} on qualitative automata is closely connected with ours on subzero automata. Firstly, some useful results from \cite{DBLP:journals/tocl/CarayolHS14} 
are exploited in this work. For example, our proof of Proposition \ref{proposition_irregularity} is similar in the spirit to the argument used in \cite{DBLP:journals/tocl/CarayolHS14} to show that there exists a language defined by a qualitative tree automaton which is not accepted by any Rabin's tree automaton. Secondly, as we show in Section \ref{section:examples}, the class of languages definable by subzero automata includes all regular languages and all languages definable by qualitative automata (cf. Proposition \ref{fullqzero}). Hence,
\begin{proposition}
Subzero automata constitute a strict generalization of both Rabin's automata and qualitative automata.
\end{proposition}
\noindent
Note that, since qualitative automata enjoy the Rabin's regularity property, our main Theorem \ref{main_result_1} implies the decidability of the emptiness problem of qualitative automata.

Thirdly, the proof of the decidability of qualitative automata uses known results from finite game-theory, namely the positional determinacy of $2\frac{1}{2}$-player turn-based parity games, and graph algorithms to solve such games. Instead, since subzero automata do not enjoy the Rabin's regularity property, our proof method is entirely different and is based on the design of a deductive system. Lastly, while the classes of zero-automata and subzero-automata have been introduced to solve the decision problem for a logic ($\textnormal{weakMSO}+\forall^{=1}_\pi$),  the class of qualitative automata does not seem to be connected to a logical theory.


\section{Technical background}

\paragraph*{Multisets.}
A multiset over a set $Q$ of elements is formally a function $w\!:\!Q\rightarrow\mathbb{N}$. We will only consider multisets over finite sets $Q$. We will use intuitive brackets notations with repetitions to denote multisets. For example, $w\!=\!\{q_1, q_1, q_2\}$ is the multiset $w$ over $\{q_1,q_2,q_3\}$ defined by $w(q_1)\!=\!2$, $w(q_2)\!=\!1$ and $w(q_3)\!=\!0$. We denote with $\subseteq$ the pointwise order on $Q\rightarrow\mathbb{N}$ and we say that $w$ is a sub(multi)set of $w^\prime$ if $w\subseteq w^\prime$. The meet (infimum) operation on the lattice  $Q\rightarrow\mathbb{N}$ is denoted by $\sqcap$. The length of $w$, denoted by $|w|$ is the sum of all multiplicities of $w$, i.e., $\displaystyle |w|=\sum_{q\in Q} w(q)$.

\begin{definition}[Maximal submultiset]\label{def:submultiset}
For a fixed finite set $Q$ and multisets $w$ and $u$, we say that $v$ is the maximal sub(multi)set of $w$ restricted to $u$, if $v=w\sqcap u$.
\end{definition}

\paragraph*{Example.} If $Q=\{q,p,r\}$ and $w=\{q,q,q,p,p,r\}$ and $u=\{ q,q,q,q,q, p\}$, then $w\sqcap u$ is equal to $\{ q,q,q, p \}$.

\paragraph*{Tree notation.} The tree automata we consider in this paper define sets of infinite binary trees with labels from a certain given alphabet. We identify a node in a  tree with a  sequence $x\in 2^*$, with $2$ denoting the set of directions $\set{0,1}$. We write $\trees \Sigma$ for the set of trees labelled by $\Sigma$, i.e.~the set of all functions $2^* \to \Sigma$. In the proofs, we will also talk about \emph{partial trees}, where the set of nodes in the domain in not necessarily all of $2^*$, but some prefix-closed subset thereof. We use standard terminology for trees: node, root, left child, right child, ancestor and descendant. In case of partial trees, we can also talk about leaves, which are nodes without any children. The following definition is standard. 
\begin{definition}\label{regular_tree_def}
A tree $t\in \trees \Sigma$ is \emph{regular} if it is representable as the infinite unfolding of a finite directed graph whose nodes are labeled by $\Sigma$. 
\end{definition}

\noindent
Equivalently, a tree $t$ is regular if and only if up-to isomorphism it has only finitely many subtrees.

\paragraph*{Probability measure over paths.} A path in the infinite binary tree can be identified with an infinite sequence in $2^\omega$ which can be also viewed as an infinite prefix-closed set of nodes that is totally ordered by the prefix relation. Given $w\!\in\!2^*$ and $\pi\!\in\! 2^\omega$ we write $w\!\leq\! \pi$ if $w$ is a finite prefix of $\pi$. The set $2^\omega$, endowed with the topology generated by the basic clopen sets $U_{w}\!=\!\{ x\in 2^{\omega} \mid w\leq x\}$, for $w\!\in\! 2^*$, is homeomorphic to the Cantor space.  We consider the \emph{coin-flipping}  complete Borel measure $\mu$ on $2^{\omega}$ uniquely determined by the assignment $\mu(U_w)\!=\!\frac{1}{2^{|w|}}$ on the basic clopen sets, where $|w|$ denotes the length of $w$. The measure $\mu$ is also known as the Lebesgue or the uniform measure. 
 Intuitively $\mu$ models the stochastic process of generation of an infinite path in the full binary tree by a sequence of coin tosses. A $\mu$-measurable subset $A\!\subseteq\!2^{\omega}$ is called $\emph{null}$ or $\emph{negligible}$ if it has measure $0$, i.e., if $\mu(A)\!=\!0$. See, e.g.,  \cite{Kechris} for a reference on the subject. 
\section{Subzero Automata}
\label{sec:zero-automata} 
In this section we define a class of tree automata generalizing ordinary nondeterministic parity automata,  called \emph{subzero automata}, which itself is a simplification of the class of \emph{zero automata} introduced in  \cite{bojanzero2016}. 
We assume some familiarity with tree automata over infinite trees as in, e.g., \cite{Thomas1997}.


\begin{definition}\label{def:zero-automata}
	A \emph{\zero} automaton consists of a tuple
	\begin{align*}
		\underbrace{Q}_{\text{states}} \qquad \underbrace{\Sigma}_{\text{input alphabet}}  \qquad \underbrace{\delta \subseteq Q \times \Sigma \times Q \times Q}_{\text{transition relation}} \qquad \underbrace{\leq \ \subseteq Q\times Q}_{\text{total order on }Q}
			\end{align*}
	with all components finite, together with two sets of states: 
$	\underbrace{\Qall\subseteq Q}_{\text{all states}},
\
	\underbrace{\Qzero\subseteq Q}_{\text{zero states}}$.
\end{definition}

The first part of the definition matches that of ordinary 
nondeterministic parity tree automata \footnote{\label{foot_1}In an ordinary nondeterministic parity tree automaton to each state is assigned a  priority. Every nondeterministic parity automaton can be transformed into an equivalent parity tree automaton such that to each state is assigned a {\em unique} priority. This unique priority determines a total order on states. We decided to use a total order in Definition \ref{def:zero-automata}, 
because this simplifies notationally our main proof --- we avoid an induction over a finite partial order in favor of an induction over a finite total order.}. The total order $\leq$ on states indicates the priority of states, with  $q_1\leq q_2$ meaning that the state $q_2$ has higher priority than $q_1$. The new aspect of the above definition is the presence of two sets of states $\Qall$ and $\Qzero$, which determine two different conditions that a run must satisfy to be accepting. The following notion of run in a $\zero$ automaton corresponds to the usual one of nondeterministic tree automata and is entirely standard.
\
\begin{definition}[Runs]
A run of a $\zero$ automaton on a tree $t\!\in\!\trees \Sigma$ is 
a labeling $\rho\!\in\! \trees Q$ of the full binary tree with states, which is consistent with the transition relation, that is, if some node $x\!\in\! 2^*$ is labeled with $q$ and has left and right children labeled by $q_0$ and $q_1$, respectively, then the automaton has a transition of the form $(q,a,(q_0,q_1))$, where $a$ is the letter labeling the vertex $x$ in $t$. 
\end{definition}
\begin{definition}[Maximal state]
Given an infinite branch $\pi$ in a run we write $\maxinf(\pi)$ for the maximal (in the order $\leq$) state appearing infinitely often in the branch.
\end{definition}
 The following is the crucial definition regarding \zero\ automata: 
 \begin{definition}[Accepting run]\label{def:accepting:run}
 A run $\rho$ is \emph{accepting} if the following two conditions hold:
\begin{enumerate}
\item $\forall \pi. ( \maxinf(\pi)\in \Qall)$, i.e., for all infinite branches $\pi$ in $\rho$, it holds that $\maxinf(\pi)\in \Qall$, 
\item $ \mu(\{\pi \mid \maxinf(\pi)\in \Qzero\}) = 0$, i.e., the probability of the set of branches $\pi$ in $\rho$, such that $\maxinf(\pi)\in \Qzero$, is $0$.
\end{enumerate}
\end{definition}
Hence a run is accepting if all of its branches satisfy the the $\Qall$ condition and only a negligible set of paths satisfies the $\Qzero$ condition.

 \begin{definition}[Acceptance of \zero\ automata]\label{def_acceptance}
A tree $t\!\in\! \trees \Sigma$ is accepted from a state $q\!\in\! Q$ of the automaton if there exists an accepting run $\rho\!\in\! \trees Q$ with the  root labeled by $q$.
\end{definition}

\section{Examples} \label{section:examples}

In this section we illustrate the notion of $\zero$ automata with a few illustrative  examples. 

\begin{example} A $\zero$ automaton with $\Qzero\!=\!\emptyset$ is just an ordinary  nondeterministic parity automaton. Indeed the second acceptance condition in Definition \ref{def_acceptance} trivializes in this case, and the set $\Qall$ can be seen as the collection of states having even priority.
\end{example}
\begin{remark*}
The above example shows that $\zero$ automata can define all regular sets of trees. This shows a difference between $\zero$ automata and qualitative automata of \cite{DBLP:journals/tocl/CarayolHS14} as the latter class can not define all regular languages (Proposition 20 in \cite{DBLP:journals/tocl/CarayolHS14}).
\end{remark*}

\begin{example} Consider the following $\zero$ automaton
\begin{align*}
		Q=\{q,\bot\} \qquad \Sigma=\{a,b\}  \qquad \delta=\{ (q,a,\bot,\bot), (q,b,q,q), (\bot,a,\bot,\bot),  (\bot,b,\bot,\bot)\}
					\end{align*}
with $\Qall\!=\!Q$ and $\Qzero= \{ q \}$. This is a deterministic automaton where $\bot$ is a sink state. Since $\Qall\!=\!Q$, the first acceptance condition in Definition \ref{def_acceptance} trivializes. Thus the language $L\!\subseteq\!\trees{\{a,b\}}$ accepted by this automaton consists of those trees $t$ such the set of branches in $t$ having only $b$'s has probability $0$.
\end{example}

The language $L$ is an interesting example of a non-regular set (Theorem 21 in \cite{DBLP:journals/tocl/CarayolHS14}) definable by a \emph{qualitative automaton} of \cite{DBLP:journals/tocl/CarayolHS14}. In fact, the above example can be generalized to a complete characterization of languages defined  by \emph{qualitative automata} in terms of \zero\ automata:
\begin{proposition}\label{fullqzero}
Qualitative tree automata of \cite{DBLP:journals/tocl/CarayolHS14} define the same languages as  \zero\ automata such that $\Qall\!=\! Q$.
\end{proposition}

Hence \zero\ automata generalize both ordinary parity nondeterministic automata and qualitative tree automata. 

Our last example is slightly more involved and will show that there exists a \zero\ automaton $\mathcal{A}$ accepting a nonempty set of trees but not accepting any regular tree (Proposition \ref{proposition_irregularity} in Introduction and proof of Theorem 21 in \cite{DBLP:journals/tocl/CarayolHS14}). We first provide the definition of the language $L_3\!\subseteq\!\trees{\{a,b\}}$ accepted by $\mathcal{A}$ and only after describe the structure of $\mathcal{A}$.

\begin{definition}
Let $L_1\subseteq \trees{\{a,b\}}$ be the set of trees over the alphabet $\Sigma=\{a,b\}$ such that from every vertex $x$ it is possible to reach a descendant vertex $y$ labeled by the letter $a$, or as an MSO formula:
$$L_1 = \{ t \mid \forall x. \exists y . \big(x\leq y \wedge a(y) \big) \}$$

Let $L_2\subseteq \trees{\{a,b\}}$ be the set of trees such that the the set of infinite paths having infinitely many occurrences of the letter $a$ has probability $0$:
$$L_2 = \big\{ t \mid \mu\big(\{ \pi \mid  \textnormal{$\pi$ has infinitely many $a$'s}\}\big)=0 \big\}$$

Lastly, let $L_3=L_1\cap L_2$.
\end{definition}

\begin{proposition}[\cite{DBLP:journals/tocl/CarayolHS14}]
The following assertions hold:
\begin{enumerate}
\item the language $L_1$ is regular, 
\item the language $L_2$ is not regular, 
\item the language $L_3$ is not regular 
and does not contain any regular tree.
\end{enumerate}
\label{prop:serre}
\end{proposition}
\begin{proof}
Clearly $L_1$ is a regular language as it is defined by the simple MSO formula provided above.

We will now show that $L_3$ is not regular. This will immediately imply that $L_2$ is not regular as well, because otherwise $L_3\!=\!L_1\cap L_2$ would also be regular since regular languages are closed under finite intersections.

To show that $L_3$ is not regular, by Rabin's regularity theorem, it suffices to prove that it is nonempty but it does not contain any regular tree. 

\begin{claim}
$L_3$ is not empty.
\end{claim}
\begin{proof}
In order to verify $L_3$ is not empty we construct a concrete tree $t\!\in\trees{\{a,b\}}$ in $L_3$. To do this, fix any mapping $f:\mathbb{N}\to\mathbb{N}$ such that $f(0)\!=\!0$ and for all $n>0$ holds $f(n) > n + \sum^{n-1}_{i=0} f(i) $. 
We say that a vertex $x\!\in\!\{0,1\}^*$ of the full binary tree belongs to the block $n$-th if its depth $|x|$ is such that $f(n)\leq |x| < f(n+1)$. Each block can be seen as a forest of finite trees (see Figure \ref{fig:l3_witness}) of depth $f(n+1)-f(n)$.
We now describe the tree $t$. For each $n$, all nodes of the $n$-th block are labeled by $b$ {\bf except} the leftmost vertices of each (finite) tree in the block (seen as a forest). Figure \ref{fig:l3_witness} illustrates this idea. Clearly $t$ is in $L_1$.


Let $E_{n}$ be the random event (on the space of infinite branches of the full binary tree) of a path having the $f(n+1)$-th vertex labeled by $a$. Then, by construction of $t$, the probability of $E_{n}$ is exactly $\frac{1}{2^{f(n+1) - f(n)}}$.

 
This implies that $\mu(E_0)+\mu(E_1)+\ldots \leq \sum_{n=0}^\infty \frac{1}{2^{f(n+1) - f(n)}}\leq \frac{1}{2} + \ldots + \frac{1}{2^n} \leq 1$. The Borel-Cantelli lemma implies that the probability of infinitely many events $E_n$ happening is $0$. Hence the probability of the set of paths having infinitely many $a$'s is $0$. Therefore $t\!\in\! L_2$ and thus $t\!\in\! L_3$.
\end{proof}


\def\prs{\tikz[scale=.65, every node/.style={scale=0.65}, baseline=1ex,shorten >=.1pt,node distance=1.8cm,on grid,semithick,auto,
every state/.style={fill=white,draw=black,circular drop shadow,inner sep=0mm,text=black},
accepting/.style ={fill=gray,text=white}]{
\node[state] (b) {$b$};
\begin{scope}[scale=.55, every node/.style={scale=0.55}, baseline=1ex,shorten >=.1pt,node distance=1.8cm, semithick,auto,
every state/.style={fill=white,draw=black,circular drop shadow,inner sep=0mm,text=black},
accepting/.style ={fill=gray,text=white}]
\node[state] (a1) [below left=1.7cm and 2.8cm  of b] {$a$};
\node[state] (b1) [below right=1.7cm and 2.8cm  of b] {$b$};
\node[state] (b2) [below left=1.7cm and 1.4cm  of b] {$b$};
\node[draw=none] (b3) [below right=1.7cm and 1.4cm  of b] {\ldots};
\node[draw=none] (b4) [below right=1.7cm and 0cm  of b] {\ldots};
\end{scope}
\draw [black,decorate,decoration=snake] (b) -- (a1);
\draw [black,decorate,decoration=snake] (b) -- (b1); 
\draw [black,decorate,decoration=snake] (b) -- (b2);
\draw [black,decorate,decoration=snake] (b) -- (b3);
\draw [black,decorate,decoration=snake] (b) -- (b4);
\sqsone; 
\sqstwo;
}
}
\def\sqsone{
\begin{scope}[scale=.45, every node/.style={scale=0.45}, baseline=1ex,shorten >=.1pt,node distance=1.8cm, semithick,auto,
every state/.style={fill=white,draw=black,circular drop shadow,inner sep=0mm,text=black},
accepting/.style ={fill=gray,text=white}]
\node[state] (a11) [below left=1.3cm and 2.2cm  of a1] {$a$};
\node[state] (b11) [below right=1.3cm and 2.2cm  of a1] {$b$};
\node[state] (b12) [below left=1.3cm and 1.1cm  of a1] {$b$};
\node[draw=none] (b13) [below right=1.3cm and 1.1cm  of a1] {\ldots};
\node[draw=none] (b14) [below right=1.3cm and 0cm  of a1] {\ldots};
\draw [black,decorate,decoration=snake] (a1) -- (a11);
\draw [black,decorate,decoration=snake] (a1) -- (b11); 
\draw [black,decorate,decoration=snake] (a1) -- (b12);
\draw [black,decorate,decoration=snake] (a1) -- (b13);
\draw [black,decorate,decoration=snake] (a1) -- (b14);
\end{scope}
}
\def\sqstwo{
\begin{scope}[scale=.45, every node/.style={scale=0.45}, baseline=1ex,shorten >=.1pt,node distance=1.8cm, semithick,auto,
every state/.style={fill=white,draw=black,circular drop shadow,inner sep=0mm,text=black},
accepting/.style ={fill=gray,text=white}]
\node[state] (a11) [below left=1.3cm and 2.2cm  of b1] {$a$};
\node[state] (b11) [below right=1.3cm and 2.2cm  of b1] {$b$};
\node[state] (b12) [below left=1.3cm and 1.1cm  of b1] {$b$};
\node[draw=none] (b13) [below right=1.3cm and 1.1cm  of b1] {\ldots};
\node[draw=none] (b14) [below right=1.3cm and 0cm  of b1] {\ldots};
\draw [black,decorate,decoration=snake] (b1) -- (a11);
\draw [black,decorate,decoration=snake] (b1) -- (b11); 
\draw [black,decorate,decoration=snake] (b1) -- (b12);
\draw [black,decorate,decoration=snake] (b1) -- (b13);
\draw [black,decorate,decoration=snake] (b1) -- (b14);
\end{scope}
}
\def\sqsthree{
\begin{scope}[scale=.35, every node/.style={scale=0.35}, baseline=1ex,shorten >=.1pt,node distance=1.8cm, semithick,auto,
every state/.style={fill=white,draw=black,circular drop shadow,inner sep=0mm,text=black},
accepting/.style ={fill=gray,text=white}]
\node[state] (q3) [below left=1.35cm and 0.4cm of q2r] {$q_1$};
\node[state] (s3) [below right=1.35cm and 0.4cm of q2r] {$q_2$};
\draw [black] (q2r) -- (q3);
\draw [black] (q2r) -- (s3); 
\end{scope}
}
\def\sqsfour{
\begin{scope}[scale=.35, every node/.style={scale=0.35}, baseline=1ex,shorten >=.1pt,node distance=1.8cm, semithick,auto,
every state/.style={fill=white,draw=black,circular drop shadow,inner sep=0mm,text=black},
accepting/.style ={fill=gray,text=white}]
\node[state] (q3) [below left=1.35cm and 0.4cm of p2] {$q_1$};
\node[state] (s3) [below right=1.35cm and 0.4cm of p2] {$q_3$};
\draw [black] (p2) -- (q3);
\draw [black] (p2) -- (s3); 
\end{scope}
}
\def\sqsfive{
\begin{scope}[scale=.35, every node/.style={scale=0.35}, baseline=1ex,shorten >=.1pt,node distance=1.8cm, semithick,auto,
every state/.style={fill=white,draw=black,circular drop shadow,inner sep=0mm,text=black},
accepting/.style ={fill=gray,text=white}]
\node[state] (q3) [below left=1.35cm and 0.4cm of q2] {$q_1$};
\node[state] (s3) [below right=1.35cm and 0.4cm of q2] {$q_1$};
\draw [black] (q2) -- (q3);
\draw [black] (q2) -- (s3); 
\end{scope}
}
\def\sqssix{
\begin{scope}[scale=.15, every node/.style={scale=0.15}, baseline=1ex,shorten >=.1pt,node distance=1.8cm, semithick,auto,
every state/.style={fill=white,draw=black,circular drop shadow,inner sep=0mm,text=black},
accepting/.style ={fill=gray,text=white}]
\node[state] (q5) [below left=1cm and 0.4cm of p3] {$p$};
\node[state] (s5) [below right=1cm and 0.4cm of p3] {$p$};
\draw [black] (p3) -- (q5);
\draw [black] (p3) -- (s5); 
\end{scope}
\node[draw=none] (s6) [below left=0.4cm and 0.35cm of s5,draw = none] {\ldots};
}
\def\sqsseven{
\begin{scope}[scale=.08, every node/.style={scale=0.08}, baseline=1ex,shorten >=.1pt,node distance=1.8cm, semithick,auto,
every state/.style={fill=white,draw=black,circular drop shadow,inner sep=0mm,text=black},
accepting/.style ={fill=gray,text=white}]
\node[state] (q6) [below left=0.7cm and 0.3cm of s5] {$q$};
\draw [black] (s5) -- (q6);
\end{scope}
\node[draw=none] (s6) [below right=0.7cm and 0.3cm of s5,draw = none] {\ldots};
\draw [black] (s5) -- (s6); 
}
\begin{figure}[H]
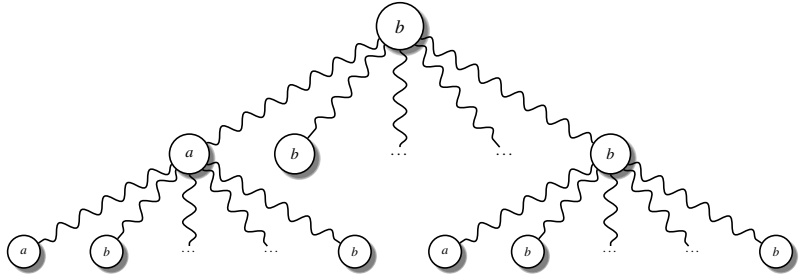

\centering
\prs
\caption{A prefix of a tree $t\in L_3$ up to the level $f(2)$.}
\label{fig:l3_witness}
\end{figure}

\begin{claim}
$L_3$ does not contain any regular tree $t$. 
\end{claim}

\begin{proof}
Indeed, let $G$ be the finite graph (where each vertex can reach exactly two vertices) representing $t$. We can view $G$ as a finite Markov chain where all edges have probability $\frac{1}{2}$. From the assumption that $t\!\in\!L_1$, we know that every vertex in $G$ can reach a vertex labeled $a$. By elementary results of Markov chains, a random infinite path in $G$ will almost surely visit infinitely many times states labeled by $a$ and this is a contradiction with the hypothesis that $t\!\in\! L_2$.
\end{proof}

\noindent
The proofs of the above two Claims finish the proof of the Proposition.
\end{proof}

Both $L_1$ and $L_2$ are easily definable by $\zero$-automata. A concrete and conveniently small $\zero$-automaton defining $L_3$ is presented below.

\begin{definition}
Let $\mathcal{A}$ be the $\zero$ automaton with $Q=\{\exists, R, \top\}$, transition relation defined as $\delta= \{(q,a,\top,\top), (q,b,\exists,R), (q,b,R,\exists) \mid q\in Q\}$, order on states $ \exists<R<\top $, $\Qall=\{ \top, R\}$ and $\Qzero=\{\top\}$.  
\end{definition}

Observe that the automaton is deterministic on reading the letter $a$ and nondeterministic on the letter $b$. Intuitively, the state $\top$ is reached exactly when the letter $a$ is read. When the letter $b$ is read, non-deterministically the automaton guesses which of the two children of the current vertex will lead to a further letter $a$ by labeling it with $\exists$ while the other child is labeled with $R$.

\begin{proposition}
The automaton $\mathcal{A}$ recognizes the language $L_3$.
\end{proposition}

\section{Decidability of the regular emptiness problem of \zero\ automata}
\label{sec:emptiness}

We define the \emph{regular emptiness problem} of \zero\ automata as follows.
\begin{definition}[Regular Emptiness Problem]
Given a $\zero$ automaton $\mathcal{A}$ decide if $\mathcal{A}$ accepts some regular tree (in the usual sense of Definition \ref{regular_tree_def}).
\end{definition}

The main result of this paper, stated as Theorem \ref{main_result_1} in the Introduction, is that the regular emptiness problem of \zero\ automata is decidable.
To prove Theorem \ref{main_result_1} we introduce in this section a deductive system whose rules depend on $\Aa$. We will then show that a regular tree is accepted by the automaton if and only if a certain assertion is derivable syntactically in the deductive system. Furthermore, we show that if a derivation exists, then a derivation of bounded depth (in the size of the automaton) exists. Therefore the derivation search space is finite and this implies that the regular emptiness problem of \zero\ automata is decidable.
We now proceed with some technical definitions needed to formulate the rules of the deductive system. We begin by introducing the notion of partial runs in \zero\ automata. Intuitively, these are like accepting runs (Definition \ref{def:accepting:run}) but can be partial trees and have leaves. In what follows we fix a generic \zero\ automaton.

\begin{definition}[Partial runs]\label{def:partial-run}
A  \emph{partial run with $n$ ports} is a partial binary tree labelled by states, with a partition of its leaves into nonempty sets $X_1,\ldots,X_n$, called \emph{ports}, subject to the following conditions: 
\begin{itemize}
	\item consistency with the transition function, i.e., if some node has state $q$ and its children have states $q_0,q_1$, then the automaton has a transition of the form $(q,a,(q_0,q_1))$, for some letter $a$;
	\item for every $X_i$, all leaves in $X_i$ are labelled by the same state which is called the \emph{type} of the port $X_i$ (for $i\neq j$ it may happen that two ports $X_i$ and $X_j$ have the same type).
	\item (all condition) every infinite path in the partial run has maxinf state in $\Qall$.
	\item (zero condition) the set of infinite paths having maxinf state in $\Qzero$ has probability $0$.
	\end{itemize}
\end{definition}

Note, by comparison with Definition \ref{def:accepting:run}, that every accepting run is also a partial run without any ports. Our proof system will manipulate statements about partial runs.

\begin{definition}[Profiles]
We define a \emph{profile}  to be an expression of the form
	$\profile p q   \{q_1, \cdots, q_n\}$,
where $p,q\in Q$ and $\{q_1,\dots, q_n\}$ is a multiset over $Q$. Hence, formally, a profile is a triple in $Q \times Q \times \N^Q $. 
\end{definition}
\def\prs{\tikz[scale=.50, every node/.style={scale=0.65}, baseline=1ex,shorten >=.1pt,node distance=1.5cm,on grid,semithick,auto,
every state/.style={fill=white,draw=black,circular drop shadow,inner sep=0mm,text=black},
accepting/.style ={fill=gray,text=white}]{
\node[state,accepting] (p) {$p$};
\begin{scope}[scale=.55, every node/.style={scale=0.55}, baseline=1ex,shorten >=.1pt,node distance=1.5cm, semithick,auto,
every state/.style={fill=white,draw=black,circular drop shadow,inner sep=0mm,text=black},
accepting/.style ={fill=gray,text=white}]
\node[state] (r) [below left=1.7cm and 0.7cm  of p] {$r$};
\node[state] (s) [below right=1.7cm and 0.7cm  of p] {$s$};
\end{scope}
\draw [black] (p) -- (r);
\draw [black] (p) -- (s); 
\sqsone; 
\rqp;
\node[draw=none] (portp) [below left=8cm and 1cm  of p] {$p$ port};
\node[draw=none] (portq) [below right=8cm and 0.5cm  of p] {$q$ port};
\node[draw=none] (portq1) [below right=8cm and 1.70cm  of p] {$q$ port};
\draw [red,dashed] (q2) -- (portq);
\draw [red,dashed] (q2r) -- (portq);
\draw [red,dashed] (q3) -- (portq);
\draw [red,dashed] (q4) -- (portq1);
\draw [red,dashed] (q5) -- (portq1);
\draw [red,dashed] (q6) -- (portq1);
\draw [red,dashed] (p2) -- (portp);
}
}
\def\sqsone{
\begin{scope}[scale=.45, every node/.style={scale=0.45}, baseline=1ex,shorten >=.1pt,node distance=1.8cm, semithick,auto,
every state/.style={fill=white,draw=black,circular drop shadow,inner sep=0mm,text=black},
accepting/.style ={fill=gray,text=white}]
\node[state] (q2r) [below left=1.5cm and 0.4cm of s] {$q$};
\node[state] (s2) [below right=1.5cm and 0.4cm of s] {$s$};
\draw [black] (s) -- (q2r);
\draw [black] (s) -- (s2); 
\sqstwo;
\sqsthree;
\sqsfour;
\sqsfive;
\end{scope}
}
\def\rqp{
\begin{scope}[scale=.45, every node/.style={scale=0.45}, baseline=1ex,shorten >=.1pt,node distance=1.8cm, semithick,auto,
every state/.style={fill=white,draw=black,circular drop shadow,inner sep=0mm,text=black},
accepting/.style ={fill=gray,text=white}]
\node[state] (q2) [below left=1.5cm and 0.4cm of r] {$q$};
\node[state] (p2) [below right=1.5cm and 0.4cm of r] {$p$};
\draw [black] (r) -- (q2);
\draw [black] (r) -- (p2); 
\end{scope}
}
\def\sqstwo{
\begin{scope}[scale=.35, every node/.style={scale=0.35}, baseline=1ex,shorten >=.1pt,node distance=1.8cm, semithick,auto,
every state/.style={fill=white,draw=black,circular drop shadow,inner sep=0mm,text=black},
accepting/.style ={fill=gray,text=white}]
\node[state] (q3) [below left=1.35cm and 0.4cm of s2] {$q$};
\node[state] (s3) [below right=1.35cm and 0.4cm of s2] {$s$};
\draw [black] (s2) -- (q3);
\draw [black] (s2) -- (s3); 
\end{scope}
}
\def\sqsthree{
\begin{scope}[scale=.25, every node/.style={scale=0.25}, baseline=1ex,shorten >=.1pt,node distance=1.8cm, semithick,auto,
every state/.style={fill=white,draw=black,circular drop shadow,inner sep=0mm,text=black},
accepting/.style ={fill=gray,text=white}]
\node[state] (q4) [below left=1.2cm and 0.4cm of s3] {$q$};
\node[state] (s4) [below right=1.2cm and 0.4cm of s3] {$s$};
\draw [black] (s3) -- (q4);
\draw [black] (s3) -- (s4); 
\end{scope}
}
\def\sqsfour{
\begin{scope}[scale=.15, every node/.style={scale=0.15}, baseline=1ex,shorten >=.1pt,node distance=1.8cm, semithick,auto,
every state/.style={fill=white,draw=black,circular drop shadow,inner sep=0mm,text=black},
accepting/.style ={fill=gray,text=white}]
\node[state] (q5) [below left=1cm and 0.4cm of s4] {$q$};
\node[state] (s5) [below right=1cm and 0.4cm of s4] {$s$};
\draw [black] (s4) -- (q5);
\draw [black] (s4) -- (s5); 
\end{scope}
}
\def\sqsfive{
\begin{scope}[scale=.08, every node/.style={scale=0.08}, baseline=1ex,shorten >=.1pt,node distance=1.8cm, semithick,auto,
every state/.style={fill=white,draw=black,circular drop shadow,inner sep=0mm,text=black},
accepting/.style ={fill=gray,text=white}]
\node[state] (q6) [below left=0.7cm and 0.3cm of s5] {$q$};
\draw [black] (s5) -- (q6);
\end{scope}
\node[draw=none] (s6) [below right=0.7cm and 0.3cm of s5,draw = none] {\ldots};
\draw [black] (s5) -- (s6); 
}
\vspace{-20pt}
\setlength{\intextsep}{25pt}%
\begin{figure}[H]
\centering
\prs
\caption{A partial run with root $p$, one port of type $p$ and two ports of type $q$.}
\label{figure_plug1}
\end{figure}

In what follows we reserve the letter $v,w$ to range over (possibly empty) multisets over $Q$ and simply write $\profile p q w$ for an arbitrary profile. We write $\max(w)$ to indicate the maximal state (with respect to the order $\leq$ of the automaton) in $w$.

\begin{definition}[Profile of a partial run]
We say that a partial run with $n$ ports $(X_1,\dots, X_n)$ has profile $\profile p q   \{q_1, \ldots , q_n\}$, if:
\begin{enumerate}
\item 
 $p$ is the state in the root, and
\item every leaf in the $i$-th port $X_i$ has type $q_i$, and
\item $q$ is the maximal state, according to the total order on states, that labels any inner state (i.e., not a leaf) of the partial run. 
\end{enumerate}
\end{definition}
Figure \ref{figure_plug1} illustrates the concept of a partial run on an example of a $\zero$ automaton with four states $\{p,q,r,s\}$. The run has profile $\profile p { \max(\{p,q,r,s\})} \{ p,q,q \}$ with one port of type $p$ and two ports of type $q$.
Note that every accepting run, which is a partial run without ports, with root labeled by $p$ has profile $\profile p {\max(Q)} \emptyset$.

\begin{definition}[Realizable partial run]\label{def:unsaturated-profile}
We say that a profile is \emph{realizable} if it is the profile of some partial run. 
\end{definition}

Accordingly, the \zero\ automaton $\mathcal{A}$ accepts some tree from a state $q_0\!\in\!Q$ if and only the profile $\profile {q_0} {\max(Q)} \emptyset$ is realizable.

\subsection{The Deductive System}

In this subsection we fix a given $\zero$ automaton $\mathcal{A}$ and define a deductive system to derive profiles from other profiles. The deductive system has one axiom rule (A), three unary derivation rules (WL), (SL) and (D), and one binary deduction rule (U), as listed in Figure \ref{fig:unsaturated-calculus}.

\begin{figure}[H]
\begin{center}

\AxiomC{}
\LeftLabel{Axiom (A):}
\RightLabel{\scriptsize if there is a transition $(p,a,q,r)$}
\UnaryInfC{$\profile p {p} {\{q, r\}}$}
\DisplayProof

\AxiomC{$\profile p p {\{p\}} \cup w$}
\LeftLabel{Weak Looping  (WL):}
\RightLabel{\scriptsize if $p\in \Qall\setminus\Qzero$}
\UnaryInfC{$\profile p p w$}
\DisplayProof

\AxiomC{$\profile p p {\{p\}}\cup w$}
\LeftLabel{Strong Looping  (SL):}
\RightLabel{\scriptsize if $p\in \Qall$ and  $w\neq \emptyset$}
\UnaryInfC{$\profile p p {w}$}
\DisplayProof

\AxiomC{$\profile p q {\{r\}\cup w}$}
\AxiomC{$\profile r {s} v$}
\LeftLabel{Unification (U):}
\BinaryInfC{$\profile p {\max(q,s,r)} {w\cup v}$}
\DisplayProof

\AxiomC{$\profile p q w\cup\{r,r\}$}
\LeftLabel{Deduplication (D):}
\UnaryInfC{$\profile p q w\cup\{r\}$}
\DisplayProof

\caption{Calculus for profiles. Variables $p,q,r,s$ range over $Q$ and variables $w,v$ range over multisets over $Q$. }
\label{fig:unsaturated-calculus}
\end{center}
\end{figure}

\vspace{-20pt}
\noindent
The crucial properties of the deductive system are formulated as the following theorem and corollary.
\begin{theorem}\label{soundness_and_completeness}The following assertions hold:
\begin{itemize}
\item[] Soundness: if a profile $\profile p q w$ is derivable in the deductive system then $\profile p q w$ is realizable by a regular tree.
\item[] Completeness and Boundedness: if a profile $\profile p q w$ is realizable by a regular tree then it is derivable in the deductive system with a derivation of size\footnote{The size of a derivation tree is defined as the number of its vertices.} smaller or equal than $f(q,|w|)$, where $f$ is a primitive recursive function $f\!:\!\mathbb{N}\times \mathbb{N}\rightarrow \mathbb{N}$ (we identify the linearly ordered set of states $Q$ with the corresponding initial segment of natural numbers $\{0,\ldots,|Q|-1\}$). 
\end{itemize}
\end{theorem}

\noindent
The proof of Theorem \ref{soundness_and_completeness}, together with the definition of $f$, is presented in Sections \ref{sec:soundness} and \ref{sec:completness}.
\begin{corollary}
The regular emptiness problem of \zero\ automata is decidable. 
\end{corollary}
\begin{proof}
For a given initial state $q_0\!\in\! Q$, the automaton $\Aa$ accepts some tree if and only if there exists some accepting run starting from $q_0$, i.e., if there exists a realizable partial run having profile $\profile {q_0} {\max \{Q\}} \emptyset$. It follows from  Theorem \ref{soundness_and_completeness} that such a run exists if and only if there exists a valid derivation of $\profile {q_0} {\max \{Q\}} \emptyset$ having size at most $f(q,|w|)$. The number of derivations of size (i.e., number of nodes) at most $f(q,|w|)$ is finite. Indeed note by inspection of each of the six rules of the derivation system that, for each vertex $\profile {p} {q} w$ in a derivation the number of valid premises is finite. Therefore the existence of a derivation of $\profile {q_0} {\max \{Q\}} \emptyset$ of size at most $f(q,|w|)$ can be checked in finite time by enumerating all such derivations. 
\end{proof}

The upper bound provided by the primitive recursive function $f\!:\!\mathbb{N}\times \mathbb{N}\rightarrow \mathbb{N}$ is by no means tight, and better upper bounds might exist. The main goal of this work has been to establish the decidability of the regular emptiness problem. The analysis of its computational complexity is an interesting topic for future work.



\section{Proof of Soundness}
\label{sec:soundness}

We need to show that if a profile is derivable then it is realizable by a regular tree.
We prove this by induction on the complexity of the derivation tree.

\paragraph*{Case (A):}

The base case is given by profiles derived by application of the axiom rule (A).
%
%
In this case, the profile is of the form $\profile p p \{q,r\}$ for some transition $(p,a,q,r)$, with $a\!\in\! \Sigma$, of the automaton. Therefore the following tree  
is a regular partial run with profile  $\profile p p \{q,r\}$.

\def\prsa{\tikz[scale=.65, every node/.style={scale=0.65}, baseline=1ex,shorten >=.1pt,node distance=1.8cm,on grid,semithick,auto,
every state/.style={fill=white,draw=black,circular drop shadow,inner sep=0mm,text=black},
accepting/.style ={fill=gray,text=white}]{
\node[state,accepting] (p1) {$p$};
\begin{scope}[scale=.55, every node/.style={scale=0.55}, baseline=1ex,shorten >=.1pt,node distance=1.8cm, semithick,auto,
every state/.style={fill=white,draw=black,circular drop shadow,inner sep=0mm,text=black},
accepting/.style ={fill=gray,text=white}]
\node[state] (r) [below left=1.7cm and 1.4cm  of p1] {q};
\node[state] (s) [below right=1.7cm and 1.4cm  of p1] {r};
\end{scope}
\draw [black] (p) -- (r);
\draw [black] (p) -- (s); 
}
}
\begin{figure}[H]
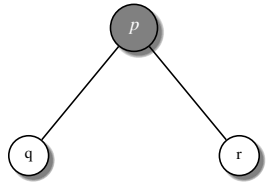

\centering
\prsa
\caption{The partial run corresponding to the axiom rule (A).}
\label{fig:axiom}
\end{figure}

\paragraph*{Case (WL):}

Assume the derivation ends with an application of the weak looping rule (WL).
\begin{center}
\AxiomC{$\profile p p \{p\}\cup w$}
\LeftLabel{Weak Looping  (WL):}
\RightLabel{\scriptsize if $p\in \Qall\setminus\Qzero$}
\UnaryInfC{$\profile p p w$}
\DisplayProof
\end{center}

By induction hypothesis, we can assume that the profile $\profile p p \{p\}\cup w$ is realizable by some regular partial run $\rho$. To obtain a realization of $\profile p p w$, we simply ``unfold'' or  ``loop'' in the partial run $\rho$ as illustrated in Figure \ref{fig:looping}. When viewing $\rho$ as a finite graph, this corresponds in adding a loop from all ports labeled by $p$ to the root $p$.

\def\prsa{\tikz[scale=.65, every node/.style={scale=0.65}, baseline=1ex,shorten >=.1pt,node distance=1.8cm,on grid,semithick,auto,
every state/.style={fill=white,draw=black,circular drop shadow,inner sep=0mm,text=black},
accepting/.style ={fill=gray,text=white}]{
\node[below left=1.4cm and 0.0cm of p,state] (p1) {$p$};
\begin{scope}[scale=.55, every node/.style={scale=0.55}, baseline=1ex,shorten >=.1pt,node distance=1.8cm, semithick,auto,
every state/.style={fill=white,draw=black,circular drop shadow,inner sep=0mm,text=black},
accepting/.style ={fill=gray,text=white}]
\node[draw=none] (r) [below left=1.7cm and 1.4cm  of p1] {};
\node[draw=none] (s) [below right=1.7cm and 1.4cm  of p1] {};
\end{scope}
\draw [black] (p) -- (r);
\draw [black] (p) -- (s); 
\node[draw=none] (q) [below left=1.4cm and 0.0cm of p1, draw = none] {};
\draw [black,decorate,decoration=snake] (p1) -- (q); 
}
}
\def\prsb{\tikz[scale=.65, every node/.style={scale=0.65}, baseline=1ex,shorten >=.1pt,node distance=1.8cm,on grid,semithick,auto,
every state/.style={fill=white,draw=black,circular drop shadow,inner sep=0mm,text=black},
accepting/.style ={fill=gray,text=white}]{
\node[state,accepting] (p) {$p$};
\begin{scope}[scale=.55, every node/.style={scale=0.55}, baseline=1ex,shorten >=.1pt,node distance=1.8cm, semithick,auto,
every state/.style={fill=white,draw=black,circular drop shadow,inner sep=0mm,text=black},
accepting/.style ={fill=gray,text=white}]
\node[draw=none] (r) [below left=1.7cm and 1.4cm  of p] {};
\node[draw=none] (s) [below right=1.7cm and 1.4cm  of p] {};
\end{scope}
\draw [black] (p) -- (r);
\draw [black] (p) -- (s); 
\node[draw=none] (q) [below left=1.4cm and 0.0cm of p, draw = none] {};
\draw [black,decorate,decoration=snake] (p) -- (q); 
\prsa
}
}
\def\prs{\tikz[scale=.65, every node/.style={scale=0.65}, baseline=1ex,shorten >=.1pt,node distance=1.8cm,on grid,semithick,auto,
every state/.style={fill=white,draw=black,circular drop shadow,inner sep=0mm,text=black},
accepting/.style ={fill=gray,text=white}]{
\node[state,accepting] (p) {$p$};
\begin{scope}[scale=.55, every node/.style={scale=0.55}, baseline=1ex,shorten >=.1pt,node distance=1.8cm, semithick,auto,
every state/.style={fill=white,draw=black,circular drop shadow,inner sep=0mm,text=black},
accepting/.style ={fill=gray,text=white}]
\node[draw=none] (r) [below left=1.7cm and 1.4cm  of p] {};
\node[draw=none] (s) [below right=1.7cm and 1.4cm  of p] {};
\end{scope}
\draw [black] (p) -- (r);
\draw [black] (p) -- (s); 
\node[draw=none] (q) [below left=1.4cm and 0.0cm of p, draw = none] {$p$};
\draw [black,decorate,decoration=snake] (p) -- (q); 
\begin{scope}[shift={(7cm,0)}]
\node[state,accepting] (p) {$p$};
\begin{scope}[scale=.55, every node/.style={scale=0.55}, baseline=1ex,shorten >=.1pt,node distance=1.8cm, semithick,auto,
every state/.style={fill=white,draw=black,circular drop shadow,inner sep=0mm,text=black},
accepting/.style ={fill=gray,text=white}]
\node[draw=none] (r) [below left=1.7cm and 1.4cm  of p] {};
\node[draw=none] (s) [below right=1.7cm and 1.4cm  of p] {};
\end{scope}
\draw [black] (p) -- (r);
\draw [black] (p) -- (s); 
\node[draw=none] (q) [below left=1.4cm and 0.0cm of p, draw = none] {};
\draw [black,decorate,decoration=snake] (p) -- (q); 
\node[below left=1.4cm and 0.0cm of p,state] (p1) {$p$};
\begin{scope}[scale=.55, every node/.style={scale=0.55}, baseline=1ex,shorten >=.1pt,node distance=1.8cm, semithick,auto,
every state/.style={fill=white,draw=black,circular drop shadow,inner sep=0mm,text=black},
accepting/.style ={fill=gray,text=white}]
\node[draw=none] (r) [below left=1.7cm and 1.4cm  of p1] {};
\node[draw=none] (s) [below right=1.7cm and 1.4cm  of p1] {};
\end{scope}
\draw [black] (p1) -- (r);
\draw [black] (p1) -- (s); 
\node[draw=none] (q) [below left=1.4cm and 0.0cm of p1, draw = none] {$p$};
\draw [black,decorate,decoration=snake] (p1) -- (q); 
\end{scope}
\begin{scope}[shift={(14cm,0)}]
\node[state,accepting] (p) {$p$};
\begin{scope}[scale=.55, every node/.style={scale=0.55}, baseline=1ex,shorten >=.1pt,node distance=1.8cm, semithick,auto,
every state/.style={fill=white,draw=black,circular drop shadow,inner sep=0mm,text=black},
accepting/.style ={fill=gray,text=white}]
\node[draw=none] (r) [below left=1.7cm and 1.4cm  of p] {};
\node[draw=none] (s) [below right=1.7cm and 1.4cm  of p] {};
\end{scope}
\draw [black] (p) -- (r);
\draw [black] (p) -- (s); 
\node[draw=none] (q) [below left=1.4cm and 0.0cm of p, draw = none] {};
\draw [black,decorate,decoration=snake] (p) -- (q); 
\node[below left=1.4cm and 0.0cm of p,state] (p1) {$p$};
\begin{scope}[scale=.55, every node/.style={scale=0.55}, baseline=1ex,shorten >=.1pt,node distance=1.8cm, semithick,auto,
every state/.style={fill=white,draw=black,circular drop shadow,inner sep=0mm,text=black},
accepting/.style ={fill=gray,text=white}]
\node[draw=none] (r) [below left=1.7cm and 1.4cm  of p1] {};
\node[draw=none] (s) [below right=1.7cm and 1.4cm  of p1] {};
\end{scope}
\draw [black] (p1) -- (r);
\draw [black] (p1) -- (s); 
\node[draw=none] (q) [below left=1.4cm and 0.0cm of p1, draw = none] {$p$};
\draw [black,decorate,decoration=snake] (p1) -- (q); 
\node[below left=1.4cm and 0.0cm of p1,state] (p2) {$p$};
\begin{scope}[scale=.55, every node/.style={scale=0.55}, baseline=1ex,shorten >=.1pt,node distance=1.8cm, semithick,auto,
every state/.style={fill=white,draw=black,circular drop shadow,inner sep=0mm,text=black},
accepting/.style ={fill=gray,text=white}]
\node[draw=none] (r) [below left=1.7cm and 1.4cm  of p2] {};
\node[draw=none] (s) [below right=1.7cm and 1.4cm  of p2] {};
\end{scope}
\draw [black] (p2) -- (r);
\draw [black] (p2) -- (s); 
\node[draw=none] (q) [below left=1.4cm and 0.0cm of p2, draw = none] {\vdots};
\draw [black,decorate,decoration=snake] (p2) -- (q); 
\end{scope}
}
}
\begin{figure}[H]
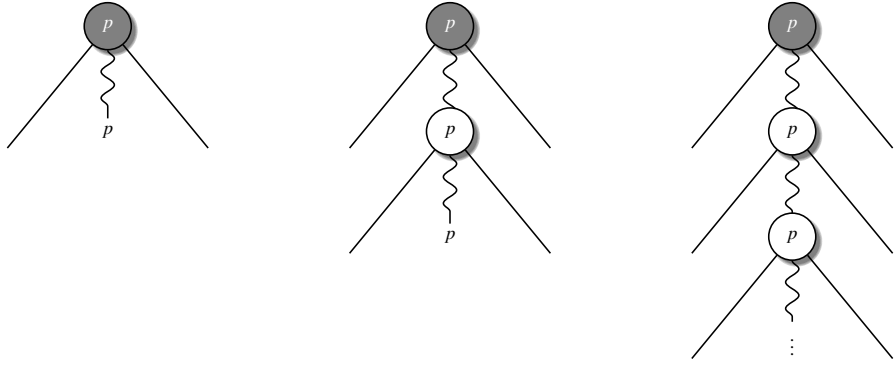

\centering
\prs
\caption{Illustration of the ($WL$) rule. From left to right, the original partial run is iteratively plugged into the port $p$.}
\label{fig:looping}
\end{figure}

\vspace{-10pt}
This construction is correct, because $p$ is the greatest (w.r.t. the order $\leq$ on  states) in any root-to-port path, and $p\in\qall$. Hence all new infinite paths generated by this unfolding have infinitely many occurrences of $p$'s and hence satisfy the acceptance condition. 
Furthermore, note that the proviso of the rule guarantees that $p\!\not\in\! \qzero$, hence there is nothing to preserve regarding the probabilistic condition.

%

\paragraph{Case (SL).}

Assume the derivation ends with an application of the weak looping rule (WL).

\AxiomC{$\profile p p \{p\}\cup w$}
\LeftLabel{Strong Looping  (SL):}
\RightLabel{\scriptsize if $p\in \Qall$ and  $w$ is nonempty}
\UnaryInfC{$\profile p p {w}$}
\DisplayProof

This case is very similar to the one just considered for the rule $(WL)$. The crucial difference is that, in this case, it is possible that $p\!\in\! \qzero$ and we need to guarantee that the set of infinite paths (with infinitely many $p$) generated by the looping construction has probability $0$.

This is enforced by the assumption $w\neq \emptyset$. This means that the partial run $\rho$ has some reachable ports. Therefore, the set of paths eventually in ending in a port has a probability $\epsilon$ {\bf strictly bigger than $0$}. This ensures that, after the looping construction, the set of infinite paths never reaching a port has probability $0$.

\def\prs{\tikz[scale=.65, every node/.style={scale=0.65}, baseline=1ex,shorten >=.1pt,node distance=1.8cm,on grid,semithick,auto,
every state/.style={fill=white,draw=black,circular drop shadow,inner sep=0mm,text=black},
accepting/.style ={fill=gray,text=white}]{
\node[state,accepting] (p) {$p$};
\begin{scope}[scale=.55, every node/.style={scale=0.55}, baseline=1ex,shorten >=.1pt,node distance=1.8cm, semithick,auto,
every state/.style={fill=white,draw=black,circular drop shadow,inner sep=0mm,text=black},
accepting/.style ={fill=gray,text=white}]
\node[draw=none] (r) [below left=1.7cm and 1.4cm  of p] {};
\node[draw=none] (s) [below right=1.7cm and 1.4cm  of p] {};
\end{scope}
\draw [black] (p) -- (r);
\draw [black] (p) -- (s); 
\node[draw=none] (q) [below left=1.4cm and 0.0cm of p, draw = none] {$r$};
\draw [black,decorate,decoration=snake] (p) -- (q); 
\begin{scope}[shift={(7cm,0)}]
\node[state,accepting] (p) {$r$};
\begin{scope}[scale=.55, every node/.style={scale=0.55}, baseline=1ex,shorten >=.1pt,node distance=1.8cm, semithick,auto,
every state/.style={fill=white,draw=black,circular drop shadow,inner sep=0mm,text=black},
accepting/.style ={fill=gray,text=white}]
\node[draw=none] (r) [below left=1.7cm and 1.4cm  of p] {};
\node[draw=none] (s) [below right=1.7cm and 1.4cm  of p] {};
\end{scope}
\draw [black] (p) -- (r);
\draw [black] (p) -- (s); 
\node[draw=none] (q) [below left=1.4cm and 0.0cm of p, draw = none] {\vdots};
\draw [black,decorate,decoration=snake] (p) -- (q); 
\end{scope}
\begin{scope}[shift={(14cm,0)}]
\node[state,accepting] (p) {$p$};
\begin{scope}[scale=.55, every node/.style={scale=0.55}, baseline=1ex,shorten >=.1pt,node distance=1.8cm, semithick,auto,
every state/.style={fill=white,draw=black,circular drop shadow,inner sep=0mm,text=black},
accepting/.style ={fill=gray,text=white}]
\node[draw=none] (r) [below left=1.7cm and 1.4cm  of p] {};
\node[draw=none] (s) [below right=1.7cm and 1.4cm  of p] {};
\end{scope}
\draw [black] (p) -- (r);
\draw [black] (p) -- (s); 
\node[draw=none] (q) [below left=1.4cm and 0.0cm of p, draw = none] {};
\draw [black,decorate,decoration=snake] (p) -- (q); 
\node[below left=1.4cm and 0.0cm of p,state] (p1) {$r$};
\begin{scope}[scale=.55, every node/.style={scale=0.55}, baseline=1ex,shorten >=.1pt,node distance=1.8cm, semithick,auto,
every state/.style={fill=white,draw=black,circular drop shadow,inner sep=0mm,text=black},
accepting/.style ={fill=gray,text=white}]
\node[draw=none] (r) [below left=1.7cm and 1.4cm  of p1] {};
\node[draw=none] (s) [below right=1.7cm and 1.4cm  of p1] {};
\end{scope}
\draw [black] (p1) -- (r);
\draw [black] (p1) -- (s); 
\node[draw=none] (q) [below left=1.4cm and 0.0cm of p1, draw = none] {\vdots};
\draw [black,decorate,decoration=snake] (p1) -- (q); 
\end{scope}
}
}
\begin{figure}[H]
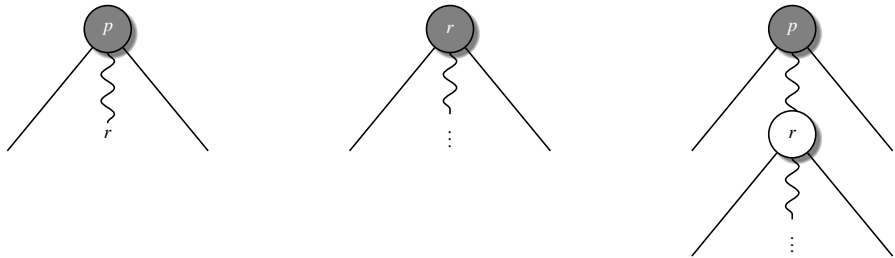

\centering
\prs
\caption{An illustration of the ($U$) rule. Into the run on the left we plug the run in the middle and obtain the run on the right.}
\label{fig:unification}
\end{figure}
\paragraph*{Case (U):}
Assume the derivation ends with an application of the binary rule (U).

\AxiomC{$\profile p q {\{r\}\cup w}$}
\AxiomC{$\profile r {s} v$}
\LeftLabel{Unification (U):}
\BinaryInfC{$\profile p {\max(q,s,r)} {w\cup v}$}
\DisplayProof

 By induction hypothesis, we can assume that the profiles $\profile p q {\{r\}\cup w}$ and $\profile r {s} v$ are realizable by regular partial runs, $\rho_1$ and $\rho_2$, respectively. To construct a regular partial run for the desired profile $\profile p {\max(q,s,r)} {w\cup v}$ it is sufficient to plug the partial run $\rho_2$ in every port $r$ of $\rho_1$, as shown in  Figure \ref{fig:unification}.

\paragraph*{Case (D):} 
This is a trivial case, as the rule simply corresponds to the operation of merging two ports having the same type.

\section{Completeness}
\label{sec:completness}

In this section we prove that if a profile $\synseq p q w$ is realizable by a regular partial run, then it is derivable in the deductive system with a derivation of size smaller or equal than $f(q,|w|)$, where 
\begin{definition} The function $f\!:\! \mathbb{N}^2\rightarrow \mathbb{N}$ is defined, by induction on the first argument, as follows:
\begin{itemize}
\item $f(0,n) = (c_1 \cdot n) + c_2$, for appropriate constants $c_1,c_2\!\in\!\mathbb{N}$, 
\item $f(q, n) = 
(K\cdot ({|Q|}+1)) + (|Q|\cdot |w|)$, where $K = \max \begin{cases}
f(q-1, 2|w|) + h(2^{|w|})\cdot |w| + |w|^2,\\
f(q-1, 2|w|)\cdot (2^{|w|} + 1) + |w|^2,\\
(3f(q-1,2) + 1)+ f(q-1,0)+1,\\
f(q-1, 2|w|) + g(2^{|w|+|Q|})\cdot {|w|} + |w|^2,
\end{cases}$
\end{itemize}
where $Q$ is the set of states of the fixed $\zero\ $ automaton and the auxiliary primitive recursive functions $g,h\!:\!\mathbb{N}\rightarrow\mathbb{N}$ are defined as:
\begin{itemize}
\item $g(0) = f(q-1, |w|+|Q|)$,
\item $g(n+1)=  f(q-1, 2(|w|+|Q|)) + g(n)\cdot (|w|+|Q|) +  (|w|+|Q|)^2$,
\end{itemize}
and
\begin{itemize}
\item $h(0) =  f(q-1, |w|)$,
\item $h(n+1) = f(q-1, 2|w|) + h(n) \cdot |w| + |w|^2$.
\end{itemize}
\label{definition_f}
\end{definition}

\begin{remark*} The formulas defining $f,g,h$ in Definition \ref{definition_f} look very technical at the first glance, but in fact they just reflect the recursive construction in the proof. 
\end{remark*}



The proof goes by induction on the order type of the state $q$ (in the well order $\leq$ on states). Since we identify the set of states $Q$ with the set of numbers $\{0,\dots, |Q|-1\}$, we write $q-1$ to denote the predecessor of $q$.

\subsection{Base case}

Let us assume that $q$ is the minimal state (w.r.t. the order $\leq$ on $Q$). That is, $q=0$. Let us fix an arbitrary partial run having profile $\synseq p q \{q_1, \ldots, q_n\}$.

In this case, since the partial run has profile $\synseq p q \{q_1, \ldots, q_n\}$, and the root $p$ is an inner vertex, we know that $p\leq q$ and therefore $q\!=\!p$. We can then conclude that the profile is of the form $\synseq p p \{q_1\ldots q_n\}$. We need to construct a derivation of $\synseq p p \{q_1\ldots q_n\}$ having size $\leq f(0,|w|)$, i.e., linearly proportional in $|w|$. Here we distinguish two cases: $p\!\in\!\Qall$ and $p\!\not\in\!\Qall$.

\vspace{-10pt}
\paragraph*{Subcase $p\not\in\qall$.}

If $p\!\not\in\Qall$ then there can not exist any infinite path in the partial run. Indeed every infinite path only consists of inner states, thus labeled by $p$.  Hence the existence of such an infinite path is a contradiction with the definition of partial run which requires that every infinite path has $\maxinf$ state in $\Qall$. 
As a consequence the partial run is well-funded and in fact, by weak K{\"o}nig's lemma, a finite  tree. A derivation of the profile $\synseq p p \{q_1,\ldots, q_n\}$ is then obtainable by subsequent applications of the rules (A) and (U) and (D).
This is best explained by a simple example.  Figure \ref{fig:base-cases-odd} illustrates the case of a partial run of profile $\synseq p p {\{q_1,q_2,q_3,q_4\}}$. 

\def\prs{\tikz[scale=.55, every node/.style={scale=0.65}, baseline=1ex,shorten >=.1pt,node distance=1.8cm,on grid,semithick,auto,
every state/.style={fill=white,draw=black,circular drop shadow,inner sep=0mm,text=black},
accepting/.style ={fill=gray,text=white}]{
\node[state,accepting] (p) {$p$};
\begin{scope}[scale=.50, every node/.style={scale=0.55}, baseline=1ex,shorten >=.1pt,node distance=1.8cm, semithick,auto,
every state/.style={fill=white,draw=black,circular drop shadow,inner sep=0mm,text=black},
accepting/.style ={fill=gray,text=white}]
\node[state] (r) [below left=1.7cm and 1.4cm  of p] {$p$};
\node[state] (s) [below right=1.7cm and 1.4cm  of p] {$p$};
\end{scope}
\draw [black] (p) -- (r);
\draw [black] (p) -- (s); 
\sqsone; 
\sqstwo;
\sqsthree;
\rqp;
\sqsfive;
}
}
\def\sqsone{
\begin{scope}[scale=.45, every node/.style={scale=0.45}, baseline=1ex,shorten >=.1pt,node distance=1.8cm, semithick,auto,
every state/.style={fill=white,draw=black,circular drop shadow,inner sep=0mm,text=black},
accepting/.style ={fill=gray,text=white}]
\node[state] (q2r) [below left=1.5cm and 0.7cm of s] {$p$};
\node[state] (s2) [below right=1.5cm and 0.7cm of s] {$p$};
\draw [black] (s) -- (q2r);
\draw [black] (s) -- (s2); 
\end{scope}
}
\def\rqp{
\begin{scope}[scale=.45, every node/.style={scale=0.45}, baseline=1ex,shorten >=.1pt,node distance=1.8cm, semithick,auto,
every state/.style={fill=white,draw=black,circular drop shadow,inner sep=0mm,text=black},
accepting/.style ={fill=gray,text=white}]
\node[state] (q2) [below left=1.5cm and 0.7cm of r] {$p$};
\node[state] (p2) [below right=1.5cm and 0.7cm of r] {$q_4$};
\draw [black] (r) -- (q2);
\draw [black] (r) -- (p2); 
\end{scope}
}
\def\sqstwo{
\begin{scope}[scale=.35, every node/.style={scale=0.35}, baseline=1ex,shorten >=.1pt,node distance=1.8cm, semithick,auto,
every state/.style={fill=white,draw=black,circular drop shadow,inner sep=0mm,text=black},
accepting/.style ={fill=gray,text=white}]
\node[state] (q3) [below left=1.35cm and 0.4cm of s2] {$q_3$};
\node[state] (p3) [below right=1.35cm and 0.4cm of s2] {$q_2$};
\draw [black] (s2) -- (q3);
\draw [black] (s2) -- (p3); 
\end{scope}
}
\def\sqsthree{
\begin{scope}[scale=.35, every node/.style={scale=0.35}, baseline=1ex,shorten >=.1pt,node distance=1.8cm, semithick,auto,
every state/.style={fill=white,draw=black,circular drop shadow,inner sep=0mm,text=black},
accepting/.style ={fill=gray,text=white}]
\node[state] (q3) [below left=1.35cm and 0.4cm of q2r] {$q_1$};
\node[state] (s3) [below right=1.35cm and 0.4cm of q2r] {$q_2$};
\draw [black] (q2r) -- (q3);
\draw [black] (q2r) -- (s3); 
\end{scope}
}
\def\sqsfour{
\begin{scope}[scale=.35, every node/.style={scale=0.35}, baseline=1ex,shorten >=.1pt,node distance=1.8cm, semithick,auto,
every state/.style={fill=white,draw=black,circular drop shadow,inner sep=0mm,text=black},
accepting/.style ={fill=gray,text=white}]
\node[state] (q3) [below left=1.35cm and 0.4cm of p2] {$q_1$};
\node[state] (s3) [below right=1.35cm and 0.4cm of p2] {$q_3$};
\draw [black] (p2) -- (q3);
\draw [black] (p2) -- (s3); 
\end{scope}
}
\def\sqsfive{
\begin{scope}[scale=.35, every node/.style={scale=0.35}, baseline=1ex,shorten >=.1pt,node distance=1.8cm, semithick,auto,
every state/.style={fill=white,draw=black,circular drop shadow,inner sep=0mm,text=black},
accepting/.style ={fill=gray,text=white}]
\node[state] (q3) [below left=1.35cm and 0.4cm of q2] {$q_1$};
\node[state] (s3) [below right=1.35cm and 0.4cm of q2] {$q_1$};
\draw [black] (q2) -- (q3);
\draw [black] (q2) -- (s3); 
\end{scope}
}
\def\sqssix{
\begin{scope}[scale=.15, every node/.style={scale=0.15}, baseline=1ex,shorten >=.1pt,node distance=1.8cm, semithick,auto,
every state/.style={fill=white,draw=black,circular drop shadow,inner sep=0mm,text=black},
accepting/.style ={fill=gray,text=white}]
\node[state] (q5) [below left=1cm and 0.4cm of p3] {$p$};
\node[state] (s5) [below right=1cm and 0.4cm of p3] {$p$};
\draw [black] (p3) -- (q5);
\draw [black] (p3) -- (s5); 
\end{scope}
\node[draw=none] (s6) [below left=0.4cm and 0.35cm of s5,draw = none] {\ldots};
}
\def\sqsseven{
\begin{scope}[scale=.08, every node/.style={scale=0.08}, baseline=1ex,shorten >=.1pt,node distance=1.8cm, semithick,auto,
every state/.style={fill=white,draw=black,circular drop shadow,inner sep=0mm,text=black},
accepting/.style ={fill=gray,text=white}]
\node[state] (q6) [below left=0.7cm and 0.3cm of s5] {$q$};
\draw [black] (s5) -- (q6);
\end{scope}
\node[draw=none] (s6) [below right=0.7cm and 0.3cm of s5,draw = none] {\ldots};
\draw [black] (s5) -- (s6); 
}
\begin{figure}[H]
\centering
\prs
\caption{A partial run realizing the profile $\synseq p p {\{q_1,q_2,q_3,q_4\}}$. The partial run is finite.}
\label{fig:base-cases-odd}
\end{figure}

\vspace{-20pt}
Note that the partial run has also profile $\synseq p p {\{q_1,q_1,q_4,q_1,q_2,q_3,q_2\}}$ by identifying each leaf with a singleton port. We will show how to derive the profile $\synseq p p {\{q_1,q_1,q_4,q_1,q_2,q_3,q_2\}}$. The desired profile $\synseq p p {\{q_1,q_2,q_3,q_4\}}$ can then be obtained by iterated applications (exactly three in this case) of the Deduplication (D) rule.
The derivation of the profile $\synseq p p {\{q_1,q_1,q_4,q_1,q_2,q_3,q_2\}}$ can be obtained by first deriving the profiles $\synseq p p {\{q_1,q_1,q_4\}}$ and $\synseq p p {\{q_1,q_2,q_3,q_2\}}$ corresponding to the left subtree and the right subtree respectively. The desired profile is then obtained by two applications of the axiom ($A$) and two applications of the ($U$) rule. All these derivations are presented in Figure \ref{fig:0_p_not_qall}. 

\vspace{-10pt}
\begin{figure}[H]
\begin{prooftree}
\scriptsize
\AxiomC{}
\LeftLabel{A}
\UnaryInfC{$\synseq p p {\{p,q_4\}}$}

\AxiomC{}
\LeftLabel{A}
\UnaryInfC{$\synseq p p {\{q_1,q_1\}}$}

\LeftLabel{U}
\BinaryInfC{$\synseq p p {\{q_1,q_1,q_4\}}$} 
\end{prooftree}
\begin{prooftree}
\scriptsize
\AxiomC{}
\LeftLabel{A}
\UnaryInfC{$\synseq p p {\{p,p\}}$}

\AxiomC{}
\LeftLabel{A}
\UnaryInfC{$\synseq p p {\{q_1,q_2\}}$}

\LeftLabel{U}
\BinaryInfC{$\synseq p p {\{p, q_1,q_2\}}$}

\AxiomC{}
\LeftLabel{A}
\UnaryInfC{$\synseq p p {\{q_3,q_2\}}$}

\LeftLabel{U}
\BinaryInfC{$\synseq p p {\{q_1,q_2,q_3,q_2\}}$}
\end{prooftree}
\begin{prooftree}
\scriptsize
\AxiomC{}
\LeftLabel{A}
\UnaryInfC{$\synseq p p {\{p,p\}}$}

\AxiomC{$\synseq p p {\{q_1,q_1,q_4\}}$}

\LeftLabel{U}
\BinaryInfC{$\synseq p p {\{p, q_1,q_1,q_4\}}$} 

\AxiomC{}
\LeftLabel{A}
\UnaryInfC{$\synseq p p {\{q_1,q_2,q_3,q_2\}}$}

\LeftLabel{U}
\BinaryInfC{$\synseq p p {\{ q_1,q_1,q_4,q_1,q_2,q_3,q_2\}}$} 
\end{prooftree}
\caption{Derivation of the profile $\synseq p p {\{ q_1,q_1,q_4,q_1,q_2,q_3,q_2\}}$.}
\label{fig:0_p_not_qall}
\end{figure}

\vspace{-20pt}
Note that all instances of the axiom rule are valid, because the partial run of Figure \ref{fig:base-cases-odd} guarantees the existence of transitions $p\stackrel{}{\to} (q_1,q_1)$, $p\stackrel{}{\to} (p,q_4)$, $p\stackrel{}{\to} (p,p)$, $p\stackrel{}{\to} (q_1,q_2)$  and $p\stackrel{}{\to} (q_3,q_2)$ in the automaton. 
Furthermore, notice that the constructed derivation has size linearly proportional to the size of the assumed partial run having profile $\synseq p q \{q_1,\dots, q_n\}$. This leads to: 

\begin{claim}
There exists a regular partial run with profile $\synseq p p \{q_1,\dots, q_n\}$ which is of size linear proportional with 
the size of the multiset $w=\{q_1,\dots, q_n\}$. 
\label{claim:size_0_q_not_qall}
\end{claim}

\begin{proof} (of the Claim)
The (finite number of) transitions in the subzero automaton can be divided into three kinds:
\begin{enumerate}[label={(\arabic*)}]
\item $p\stackrel{}{\to} (p,p)$,
\item $p\stackrel{}{\to} (p,r)$ or $p\stackrel{}{\to} (r,p)$, with $r\neq p$, and
\item $p\stackrel{}{\to} (r,s)$, with $r,s\neq p$.
\end{enumerate}

Note that transition of types (2) and (3) introduce leaves in the partial run. 
\begin{definition} 
We say that a transition of type (2)  is $w$-productive if $r$ is in $w$. Similarly, we say that transition of type (3) is  $w$-productive if both $r$ and $s$ are in $w$. 
\end{definition}

Note that only transitions of type $(1)$ and $w$-productive transitions of type $(2)$ and $(3)$ can appear in a partial run having profile $\synseq p q w$, as all other transition introduce ports which are not in $w$. It is also evident that, for each $q\!\in\! w$, some $w$-productive transition introducing a leaf labeled by $q$ must appear in the partial run.

It is then clear that if a partial run having profile $\synseq p q \{q_1,\dots, q_n\}$ exists, then there exists a partial run which uses at most $|w|$ many $w$-productive transitions, introducing each of the states appearing in the multiset $w$. Hence the size of this partial run is linear in the size of $w$.  
\end{proof}

\paragraph*{Subcase $p\in\qall$.}

Also in this case all inner states are $p$'s, but since $p\in\qall$, this time a partial run of the profile $\synseq p p \{q_1\ldots q_n\}$ may contain infinite paths. Note, furthermore, that if $p\!\in\!\Qzero$, then the sequence $q_1\ldots q_n$ is not empty (i.e., there are some ports) because otherwise all paths (and, therefore, a set of probability $1$ in the full binary tree) would be in $\Qzero$.

To deal with this case, it is sufficient to find some height $h$ such that all ports in $\{q_1,\dots, q_n\}$ appear (with the required multiplicities) as leaves at some depth $\leq h$. By ``cutting'' the (potentially infinite) partial run of profile  $\synseq p p \{q_1\ldots q_n\}$ at depth $h$ we obtain a partial run having profile $\synseq p p \{q_1,\ldots, q_n, p\}$ (the new port $p$ is present if $p$ appears as some leaf at depth $\leq h$).
As an illustrative example, consider the partial run with profile $\synseq p p \{q_1, q_1, q_2, q_3\}$ of Figure \ref{fig:base-cases}. 


\def\prs{\tikz[scale=.55, every node/.style={scale=0.65}, baseline=1ex,shorten >=.1pt,node distance=1.8cm,on grid,semithick,auto,
every state/.style={fill=white,draw=black,circular drop shadow,inner sep=0mm,text=black},
accepting/.style ={fill=gray,text=white}]{
\node[state,accepting] (p) {$p$};
\begin{scope}[scale=.55, every node/.style={scale=0.55}, baseline=1ex,shorten >=.1pt,node distance=1.8cm, semithick,auto,
every state/.style={fill=white,draw=black,circular drop shadow,inner sep=0mm,text=black},
accepting/.style ={fill=gray,text=white}]
\node[state] (r) [below left=1.7cm and 1.4cm  of p] {$p$};
\node[state] (s) [below right=1.7cm and 1.4cm  of p] {$p$};
\end{scope}
\draw [black] (p) -- (r);
\draw [black] (p) -- (s); 
\sqsone; 
\sqstwo;
\sqsthree;
\sqssix;
\rqp;
\sqsfive;
\node[draw=none] (portq11) [below left=6.5cm and 2.5cm  of p] {$q_1$ port};
\node[draw=none] (portq1) [below left=6.5cm and 0.7cm  of p] {$q_1$ port};
\node[draw=none] (portq2) [below right=6.5cm and 0.5cm  of p] {$q_2$ port};
\node[draw=none] (portq3) [below right=6.5cm and 1.70cm  of p] {$q_3$ port};
\draw [red,dashed] (p2) -- (portq1);
\draw [red,dashed] (q3) -- (portq3);
\draw [red,dashed] (s3) -- (portq2);
\draw [red,dashed] (q33) -- (portq11);
}
}
\def\sqsone{
\begin{scope}[scale=.45, every node/.style={scale=0.45}, baseline=1ex,shorten >=.1pt,node distance=1.8cm, semithick,auto,
every state/.style={fill=white,draw=black,circular drop shadow,inner sep=0mm,text=black},
accepting/.style ={fill=gray,text=white}]
\node[state] (q2r) [below left=1.5cm and 0.7cm of s] {$p$};
\node[state] (ps2) [below right=1.5cm and 0.7cm of s] {$p$};
\draw [black] (s) -- (q2r);
\draw [black] (s) -- (ps2); 
\end{scope}
}
\def\rqp{
\begin{scope}[scale=.45, every node/.style={scale=0.45}, baseline=1ex,shorten >=.1pt,node distance=1.8cm, semithick,auto,
every state/.style={fill=white,draw=black,circular drop shadow,inner sep=0mm,text=black},
accepting/.style ={fill=gray,text=white}]
\node[state] (pq2) [below left=1.5cm and 0.7cm of r] {$p$};
\node[state] (p2) [below right=1.5cm and 0.7cm of r] {$q_1$};
\draw [black] (r) -- (pq2);
\draw [black] (r) -- (p2); 
\end{scope}
}
\def\sqstwo{
\begin{scope}[scale=.35, every node/.style={scale=0.35}, baseline=1ex,shorten >=.1pt,node distance=1.8cm, semithick,auto,
every state/.style={fill=white,draw=black,circular drop shadow,inner sep=0mm,text=black},
accepting/.style ={fill=gray,text=white}]
\node[state] (q3) [below left=1.35cm and 0.4cm of s2] {$q_3$};
\node[state] (p3) [below right=1.35cm and 0.4cm of s2] {$p$};
\draw [black] (s2) -- (q3);
\draw [black] (s2) -- (p3); 
\end{scope}
}
\def\sqsthree{
\begin{scope}[scale=.35, every node/.style={scale=0.35}, baseline=1ex,shorten >=.1pt,node distance=1.8cm, semithick,auto,
every state/.style={fill=white,draw=black,circular drop shadow,inner sep=0mm,text=black},
accepting/.style ={fill=gray,text=white}]
\node[state] (pq3) [below left=1.35cm and 0.4cm of q2r] {$p$};
\node[state] (s3) [below right=1.35cm and 0.4cm of q2r] {$q_2$};
\draw [black] (q2r) -- (pq3);
\draw [black] (q2r) -- (s3); 
\end{scope}
}
\def\sqsfive{
\begin{scope}[scale=.35, every node/.style={scale=0.35}, baseline=1ex,shorten >=.1pt,node distance=1.8cm, semithick,auto,
every state/.style={fill=white,draw=black,circular drop shadow,inner sep=0mm,text=black},
accepting/.style ={fill=gray,text=white}]
\node[state] (q33) [below left=1.35cm and 0.4cm of q2] {$q_1$};
\node[state] (ps3) [below right=1.35cm and 0.4cm of q2] {$p$};
\draw [black] (q2) -- (q33);
\draw [black] (q2) -- (ps3); 
\end{scope}
}
\def\sqssix{
\begin{scope}[scale=.15, every node/.style={scale=0.15}, baseline=1ex,shorten >=.1pt,node distance=1.8cm, semithick,auto,
every state/.style={fill=white,draw=black,circular drop shadow,inner sep=0mm,text=black},
accepting/.style ={fill=gray,text=white}]
\node[state] (pq5) [below left=1cm and 0.4cm of p3] {$p$};
\node[state] (s5) [below right=1cm and 0.4cm of p3] {$p$};
\draw [black] (p3) -- (pq5);
\draw [black] (p3) -- (s5); 
\end{scope}
\node[draw=none] (s6) [below left=0.4cm and 0.35cm of s5,draw = none] {\ldots};
}
\begin{figure}[H]
\centering
\prs
\caption{A partial run with profile $\synseq p p {\{q_1,q_1,q_2,q_3\}}$.}
\label{fig:base-cases}
\end{figure}

The least depth at which all ports appear with the required multiplicities is $h\!=\!3$. 
Using the rules (A), (U) and (D), as done in the previous case, we can derive the profile $\synseq p p \{q_1, q_1, q_2, q_3, p\}$ and using similar ideas  
as in Claim \ref{claim:size_0_q_not_qall} we can prove that this derivation has size linear in $|w|+1$.  

Once a derivation of the profile $\synseq p p \{q_1,q_1,q_2,q_3,p\}$ is obtained, having size linearly proportional to $|w|+1$, we derive the profile $\synseq p p \{p,q_1,q_1,q_2,q_3\}$ by applying one of the two looping rules (WL) or (SL) depending if $p\in\Qzero$ or not: 
\begin{enumerate}
\item Case $p\!\not\in \Qzero$: in this case we obtain the desired profile by application of the weak looping rule:
\begin{prooftree}
\AxiomC{$\synseq p p {\{q_1,q_1,q_2,q_3, p\}}$}
\LeftLabel{WL}
\UnaryInfC{$\synseq p p {\{q_1,q_1,q_2,q_3\}}$} 
\end{prooftree}
\item Case $p\in \Qzero$: if 
the multiset $w$ is not empty (this is the case in this example, $w = \{q_1,q_1,q_2,q_3\}$), one can apply the strong looping rule:
\begin{prooftree}
\AxiomC{$\synseq p p {\{p, q_1,q_1,q_2,q_3\}}$}
\LeftLabel{SL}
\UnaryInfC{$\synseq p p {\{q_1,q_1,q_2,q_3\}}$} 
\end{prooftree}
\end{enumerate}
Hence, the final derivation has size linearly proportional to $|w|$, i.e., size $f(0,|w|)\!=\! c_1\cdot |w| + c_2$ for appropriate constants $c_1,c_2\in\mathbb{N}$, as stated in Definition \ref{definition_f}.

\subsection{Inductive step}
\label{subsection_inductive_step}

We reason by induction on $q$ and in the inductive step we assume that any realizable profile $\synsneq p {q} \{q_1\ldots q_n\}$ is derivable in the deductive system (note the strict inequality ${<\!q}$) by a proof of size at most $f(q-1, |\{q_1,\dots, q_n\}|)$.
In order to complete the inductive step we will consider an arbitrary realizable profile of the form  $\synseq p q {\{q_1,\ldots, q_n\}}$ and we will prove that it is derivable by a derivation of size at most $f(q, |\{q_1,\dots, q_n\}|)$. 
Let us fix an arbitrary partial run $\rho$ having profile $\synseq p q {\{q_1\ldots q_n\}}$. 

\begin{figure}[t!]
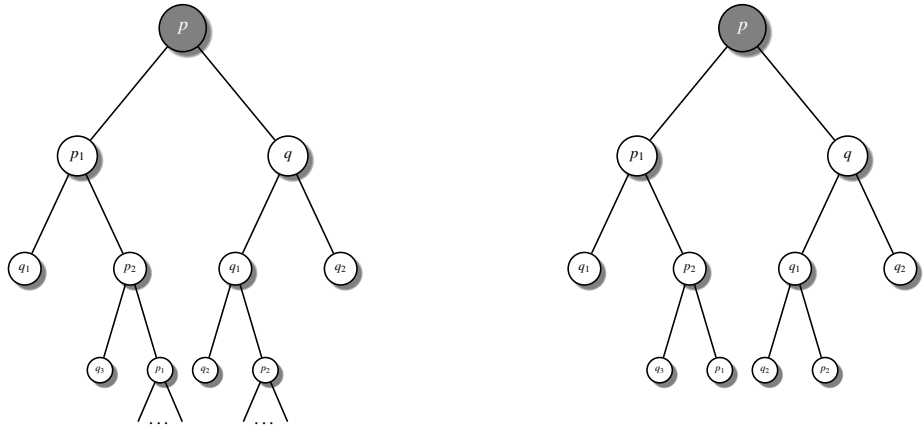

\centering
\input{pics/figure_72_1}
~
\input{pics/figure_72_2}
\caption{A run and its cut.}
\end{figure}



Let $h\!\in\!\mathbb{N}$ the least depth such that all ports in $w$ appear with at least the required multiplicities in the partial run $\rho$ at depth $\leq h$. E.g., if $w=\{r,r,s\}$ then at depth $\leq h$, there must be at least two leaves labeled by $r$ and at least one leaf labeled by $s$.  For the run $\rho$ in Figure \ref{fig:matteo_72_1} the least such $h$ is $3$ as visible in Figure \ref{fig:matteo_72_2}.
Let $\rho_h$ be the partial run $\rho$ up-to depth $h$, i.e., where all vertices below depth $h$ are removed.

Note, that $\rho_h$ has profile
$\synseq p q {w\cup P}$ where $P\subseteq Q$ is the set (i.e, multiset with multiplicities $1$) of states labeling leaves in $\rho_h$ that do not appear in $w$. 
In the example depicted above, $P=\{p_1,p_2\}$.

We can assume without loss of generality that, for each $p_i\!\in\! P$, all sub-partial runs of $\rho$ rooted at the vertices (which are leaves of $\rho_h$) labeled by $p_i$ are identical. 
If they are not, simply select one and replace all others by it. This produces another partial run $\rho^\prime$ having the same profile as $\rho$.

\def\prs{\tikz[scale=.65, every node/.style={scale=0.65}, baseline=1ex,shorten >=.1pt,node distance=1.8cm,on grid,semithick,auto,
every state/.style={fill=white,draw=black,circular drop shadow,inner sep=0mm,text=black},
accepting/.style ={fill=gray,text=white}]{
\node[state,accepting] (p) {$p_i$};
\begin{scope}[scale=.55, every node/.style={scale=0.55}, baseline=1ex,shorten >=.1pt,node distance=1.8cm, semithick,auto,
every state/.style={fill=white,draw=black,circular drop shadow,inner sep=0mm,text=black},
accepting/.style ={fill=gray,text=white}]
\node[draw=none] (r) [below left=1.7cm and 1.4cm  of p] {};
\node[draw=none] (s) [below right=1.7cm and 1.4cm  of p] {};
\end{scope}
\draw [black] (p) -- (r);
\draw [black] (p) -- (s); 
\node[draw=none] (q2m) [below left=1.4cm and 0.0cm of p, draw = none] {\ldots};
}
}
\def\sqsone{
\begin{scope}[scale=.45, every node/.style={scale=0.45}, baseline=1ex,shorten >=.1pt,node distance=1.8cm, semithick,auto,
every state/.style={fill=white,draw=black,circular drop shadow,inner sep=0mm,text=black},
accepting/.style ={fill=gray,text=white}]
\node[state] (q2r) [below left=1.5cm and 0.7cm of s] {$q_1$};
\node[state] (s2) [below right=1.5cm and 0.7cm of s] {$q_2$};
\draw [black] (s) -- (q2r);
\draw [black] (s) -- (s2); 
\end{scope}
}
\def\rqp{
\begin{scope}[scale=.45, every node/.style={scale=0.45}, baseline=1ex,shorten >=.1pt,node distance=1.8cm, semithick,auto,
every state/.style={fill=white,draw=black,circular drop shadow,inner sep=0mm,text=black},
accepting/.style ={fill=gray,text=white}]
\node[state] (q2) [below left=1.5cm and 0.7cm of r] {$q_1$};
\node[state] (p2) [below right=1.5cm and 0.7cm of r] {$p_2$};
\draw [black] (r) -- (q2);
\draw [black] (r) -- (p2); 
\end{scope}
}
\def\sqstwo{
\begin{scope}[scale=.35, every node/.style={scale=0.35}, baseline=1ex,shorten >=.1pt,node distance=1.8cm, semithick,auto,
every state/.style={fill=white,draw=black,circular drop shadow,inner sep=0mm,text=black},
accepting/.style ={fill=gray,text=white}]
\node[state] (q3) [below left=1.35cm and 0.4cm of s2] {$q_3$};
\node[state] (p3) [below right=1.35cm and 0.4cm of s2] {$p_1$};
\draw [black] (s2) -- (q3);
\draw [black] (s2) -- (p3); 
\end{scope}
}
\def\sqsthree{
\begin{scope}[scale=.35, every node/.style={scale=0.35}, baseline=1ex,shorten >=.1pt,node distance=1.8cm, semithick,auto,
every state/.style={fill=white,draw=black,circular drop shadow,inner sep=0mm,text=black},
accepting/.style ={fill=gray,text=white}]
\node[state] (q3) [below left=1.35cm and 0.4cm of q2r] {$q_2$};
\node[state] (s3p) [below right=1.35cm and 0.4cm of q2r] {$p_2$};
\draw [black] (q2r) -- (q3);
\draw [black] (q2r) -- (s3p); 
\end{scope}
}
\def\sqsfour{
\begin{scope}[scale=.35, every node/.style={scale=0.35}, baseline=1ex,shorten >=.1pt,node distance=1.8cm, semithick,auto,
every state/.style={fill=white,draw=black,circular drop shadow,inner sep=0mm,text=black},
accepting/.style ={fill=gray,text=white}]
\node[state] (q3) [below left=1.35cm and 0.4cm of p2] {$q_3$};
\node[state] (s3) [below right=1.35cm and 0.4cm of p2] {$p_1$};
\draw [black] (p2) -- (q3);
\draw [black] (p2) -- (s3); 
\end{scope}
}
\def\sqsfive{
\begin{scope}[scale=.35, every node/.style={scale=0.35}, baseline=1ex,shorten >=.1pt,node distance=1.8cm, semithick,auto,
every state/.style={fill=white,draw=black,circular drop shadow,inner sep=0mm,text=black},
accepting/.style ={fill=gray,text=white}]
\node[state] (q3) [below left=1.35cm and 0.4cm of q2] {$q_1$};
\node[state] (s3) [below right=1.35cm and 0.4cm of q2] {$p$};
\draw [black] (q2) -- (q3);
\draw [black] (q2) -- (s3); 
\end{scope}
}
\def\sqssix{
\begin{scope}[scale=.15, every node/.style={scale=0.15}, baseline=1ex,shorten >=.1pt,node distance=1.8cm, semithick,auto,
every state/.style={fill=white,draw=black,circular drop shadow,inner sep=0mm,text=black},
accepting/.style ={fill=gray,text=white}]
\node[state] (q5) [below left=1cm and 0.4cm of p3] {$p$};
\node[state] (s5) [below right=1cm and 0.4cm of p3] {$p$};
\draw [black] (p3) -- (q5);
\draw [black] (p3) -- (s5); 
\end{scope}
\node[draw=none] (s6) [below left=0.4cm and 0.35cm of s5,draw = none] {\ldots};
}
\def\sqsseven{
\begin{scope}[scale=.08, every node/.style={scale=0.08}, baseline=1ex,shorten >=.1pt,node distance=1.8cm, semithick,auto,
every state/.style={fill=white,draw=black,circular drop shadow,inner sep=0mm,text=black},
accepting/.style ={fill=gray,text=white}]
\node[draw=none] (q6) [below left=0.7cm and 0.3cm of s3] {};
\draw [black] (s3) -- (q6);
\node[draw=none] (s6) [below right=0.7cm and 0.3cm of s3,draw = none] {};
\draw [black] (s3) -- (s6); 
\end{scope}
\node[draw=none] (q6m) [below left=0.7cm and 0.0cm of s3, draw = none] {\ldots};
}
\def\sqseight{
\begin{scope}[scale=.08, every node/.style={scale=0.08}, baseline=1ex,shorten >=.1pt,node distance=1.8cm, semithick,auto,
every state/.style={fill=white,draw=black,circular drop shadow,inner sep=0mm,text=black},
accepting/.style ={fill=gray,text=white}]
\node[draw=none] (q8) [below left=0.7cm and 0.3cm of s3p] {};
\draw [black] (s3p) -- (q8);
\node[draw=none] (s8) [below right=0.7cm and 0.3cm of s3p,draw = none] {};
\draw [black] (s3p) -- (s8); 
\end{scope}
\node[draw=none] (q8m) [below left=0.7cm and 0.0cm of s3p, draw = none] {\ldots};
}
\begin{figure}[H]
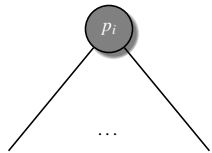

\centering
\prs
\caption{Sub-partial run rooted at $p_i$.}
\label{fig:matteo2-2}
\end{figure}



So, for each $p_i\!\in\! P$, the sub-partial run rooted at $p_i$ (Figure \ref{fig:matteo2-2}) which in what follows we denote by $\rho_{i}$, has profile $\synseq {p_i} q {w_i}$ for some $w_i\subseteq w$.
In the rest of the proof we will show how to construct derivations for:
\begin{itemize}
\item $\synseq p q {w\cup P}$, i.e., the profile of $\rho_h$, 
\item $\synseq {p_i} q {w_i}$, i.e., the profile of each $\rho_{i}$, for each $p_i\!\in\! P$.
\end{itemize}
By iterated applications ($|P|$-many) of the unification rule $(U)$, and subsequent applications of the rule $(D)$, it will then be possible to obtain a derivation of the desired profile $\synseq p q {w}$. The following figure shows one instance of such proof. 
\begin{figure}[H]
\begin{prooftree}
\AxiomC{$\synseq p q {\{q_1,q_2,q_2,q_3, p_1,p_2\}}$}
\AxiomC{$\synseq {p_1} q {w_1}$}
\LeftLabel{U}
\BinaryInfC{$\synseq p q {\{q_1,q_1,q_2,q_3, p_2\}\cup w_1}$}
\LeftLabel{D ($|w_1|$-many times)}
\UnaryInfC{$\synseq p q {\{q_1,q_1,q_2,q_3, p_2\}}$}
\AxiomC{$\synseq {p_2} q {w_2}$}
\LeftLabel{U}
\BinaryInfC{$\synseq p q {\{q_1,q_1,q_2,q_3\}\cup w_2}$}
\LeftLabel{D ($|w_2|$-many times)}
\UnaryInfC{$\synseq p q {\{q_1,q_1,q_2,q_3\}}$} 
\end{prooftree}
\caption{Example of a proof for $w=\{q_1,q_2,q_2,q_3\}$ and $P=\{p_1,p_2\}$, see Figure \ref{fig:matteo_72_1}.}
\label{proof_u_d}
\end{figure}
Furthermore, if $\synseq p q {w\cup P}$ and $\synseq {p_i} q {w_i}$ are derivable by proofs of size $N$ and $M_i$, respectively, then the final proof of 
$\synseq p q {w}$ will have size linearly proportional to $N+ \sum_{p_i\!\in\! P}M_i + |P|\cdot |w|$, where the expression $|P|\cdot |w|$ is an upper bound on the number of vertices corresponding to the application of the $(U)$ and $(D)$ rules. Therefore, by letting $K=\max{\{N, M_i\}}$ and since $|P|\leq |Q|$ where $Q$ is the set of states of the subzero automaton,  the size of derivation is bounded by 
\begin{equation}
f(q,|w|) = (K\cdot ({|Q|}+1)) + (|Q|\cdot |w|).
\label{main_equation_inductive}
\end{equation}

On the next few pages (see Equation \ref{final_inequality_last} at the very end of this Section) we will prove an upper bound 
$$ K \leq \max \begin{cases}
f(q-1, 2|w|) + h(2^{|w|})\cdot |w| + |w|^2,\\
f(q-1, 2|w|)\cdot (2^{|w|} + 1) + |w|^2,\\
(3f(q-1,2) + 1)+ f(q-1,0)+1,\\
f(q-1, 2|w|) + g(2^{|w|+|Q|})\cdot {|w|} + |w|^2.
\end{cases}$$
This agrees with  Definition \ref{definition_f} and will show that profile $\synseq p q {w}$ can be proved by a derivation of size $\leq f(q,|w|)$.

\begin{remark*} Note how the multiset $w_i\subseteq w$ of each profile $\synseq {p_i} q {w_i}$ is irrelevant in the process of proof construction, as all the ports in $w_i$ are removed by $|w_i|$-many applications of the (D) rule immediately after being introduced by the (U) rule.
Hence deriving a profile $\synseq {p_i} q {u}$, for  some $u\subseteq w$, is sufficient for the purpose.
\end{remark*}

\paragraph{How to construct a derivation of: $\synseq p q {w\cup P}$.\\}

Since $\synseq p q {w\cup P}$ is the profile of the (finite) partial run $\rho_h$, then the profile is derivable by iterated applications of the rules $(U)$, $(A)$ and $(D)$, by using the same ideas discussed in the base case of the induction (see Figure \ref{fig:0_p_not_qall}). 

What will require more work is to establish an upper bound $N$ on the size of this derivation. The reader not interested in the details regarding the estimation of the value $N$ can ignore the following claim and its proof. 

\begin{claim} The profile $\synseq p q {w\cup P}$ can be derived by a derivation of size
$$N \leq f(q-1, 2|w|) + (g(2^{|w|+|Q|})\cdot {|w|}) + |w|^2.$$
where the function $g\!:\!\mathbb{N}\rightarrow\mathbb{N}$ is defined as in Definition \ref{definition_f}.
\label{claim:ind_q_not_qall}
\label{claim:frontier_estimation}
\end{claim}

\begin{proof}[Proof of the claim.]

For each inner vertex $x$ (i.e., not a leaf) labeled by $q$ in $\rho_h$, let $\rho^x_h$ be the sub-partial run rooted at $x$. The following picture shows $\rho_h^x$ where $x$ is the right child of the root of the partial run $\rho_h$ showed in Figure \ref{fig:matteo_72_2}.
\def\prs{\tikz[scale=.65, every node/.style={scale=0.65}, baseline=1ex,shorten >=.1pt,node distance=1.8cm,on grid,semithick,auto,
every state/.style={fill=white,draw=black,circular drop shadow,inner sep=0mm,text=black},
accepting/.style ={fill=gray,text=white}]{
\node[state] (s) {$q$};
\sqsone; 
\sqsthree;
}
}
\def\sqsone{
\begin{scope}[scale=.45, every node/.style={scale=0.45}, baseline=1ex,shorten >=.1pt,node distance=1.8cm, semithick,auto,
every state/.style={fill=white,draw=black,circular drop shadow,inner sep=0mm,text=black},
accepting/.style ={fill=gray,text=white}]
\node[state] (q2r) [below left=1.5cm and 0.7cm of s] {$q_1$};
\node[state] (s2) [below right=1.5cm and 0.7cm of s] {$q_2$};
\draw [black] (s) -- (q2r);
\draw [black] (s) -- (s2); 
\end{scope}
}
\def\rqp{
\begin{scope}[scale=.45, every node/.style={scale=0.45}, baseline=1ex,shorten >=.1pt,node distance=1.8cm, semithick,auto,
every state/.style={fill=white,draw=black,circular drop shadow,inner sep=0mm,text=black},
accepting/.style ={fill=gray,text=white}]
\node[state] (q2) [below left=1.5cm and 0.7cm of r] {$q_1$};
\node[state] (p2) [below right=1.5cm and 0.7cm of r] {$p_2$};
\draw [black] (r) -- (q2);
\draw [black] (r) -- (p2); 
\end{scope}
}
\def\sqsthree{
\begin{scope}[scale=.35, every node/.style={scale=0.35}, baseline=1ex,shorten >=.1pt,node distance=1.8cm, semithick,auto,
every state/.style={fill=white,draw=black,circular drop shadow,inner sep=0mm,text=black},
accepting/.style ={fill=gray,text=white}]
\node[state] (q3) [below left=1.35cm and 0.4cm of q2r] {$q_2$};
\node[state] (s3p) [below right=1.35cm and 0.4cm of q2r] {$p_2$};
\draw [black] (q2r) -- (q3);
\draw [black] (q2r) -- (s3p); 
\end{scope}
}
\def\sqsfour{
\begin{scope}[scale=.35, every node/.style={scale=0.35}, baseline=1ex,shorten >=.1pt,node distance=1.8cm, semithick,auto,
every state/.style={fill=white,draw=black,circular drop shadow,inner sep=0mm,text=black},
accepting/.style ={fill=gray,text=white}]
\node[state] (q3) [below left=1.35cm and 0.4cm of p2] {$q_3$};
\node[state] (s3) [below right=1.35cm and 0.4cm of p2] {$p_1$};
\draw [black] (p2) -- (q3);
\draw [black] (p2) -- (s3); 
\end{scope}
}
\def\sqsseven{
\begin{scope}[scale=.08, every node/.style={scale=0.08}, baseline=1ex,shorten >=.1pt,node distance=1.8cm, semithick,auto,
every state/.style={fill=white,draw=black,circular drop shadow,inner sep=0mm,text=black},
accepting/.style ={fill=gray,text=white}]
\node[draw=none] (q6) [below left=0.7cm and 0.3cm of s3] {};
\draw [black] (s3) -- (q6);
\node[draw=none] (s6) [below right=0.7cm and 0.3cm of s3,draw = none] {};
\draw [black] (s3) -- (s6); 
\end{scope}
\node[draw=none] (q6m) [below left=0.7cm and 0.0cm of s3, draw = none] {\ldots};
}
\def\sqseight{
\begin{scope}[scale=.08, every node/.style={scale=0.08}, baseline=1ex,shorten >=.1pt,node distance=1.8cm, semithick,auto,
every state/.style={fill=white,draw=black,circular drop shadow,inner sep=0mm,text=black},
accepting/.style ={fill=gray,text=white}]
\node[draw=none] (q8) [below left=0.7cm and 0.3cm of s3p] {};
\draw [black] (s3p) -- (q8);
\node[draw=none] (s8) [below right=0.7cm and 0.3cm of s3p,draw = none] {};
\draw [black] (s3p) -- (s8); 
\end{scope}
\node[draw=none] (q8m) [below left=0.7cm and 0.0cm of s3p, draw = none] {\ldots};
}
\begin{figure}[H]
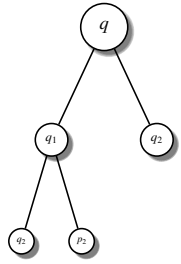

\centering
\prs
\caption{Partial run $\rho^x_h$.}
\label{fig:matteo_72_4}
\end{figure}

We define the \emph{type} of $x$ as the maximal multiset (see Definition \ref{def:submultiset}) $w_x\subseteq w\cup P$ with respect to the multiset of states labeling leaves in $\rho^x_h$. In particular $\rho^x_h$ has profile $\synseq q q {w_x}$. Note that there are only finitely many types, as there are only finitely many multisubsets of $w_x\subseteq w\cup P$. 
We now show how to construct derivations of each of the finitely many profiles $\synseq q q {w_x}$.
The proof is by induction on the \emph{complexity} of types, which we now define.

\begin{definition}
A type $v\subseteq w\cup P$ is of \emph{complexity $0$} if there exists some $x$ (labeled by $q$) in $\rho_h$ such that:
\begin{enumerate}
\item there are no vertices labeled by $q$ in any path from $x$ (excluded) to any leaf (also excluded) in $\rho_h$,
\item $x$ has type $v$, that is $v=w_x$.
\end{enumerate}
A type $v\subseteq w\cup P$ has complexity $n+1$ if there exists some $x$ (labeled by $q$) in $\rho_h$ such that:
\begin{enumerate}
\item the first state $y$ labeled by $q$ in any path from $x$ (excluded) to any leaf (also excluded) has type of complexity $\leq n$,
\item $x$ has type $v$, that is $v=w_x$.
\end{enumerate}
\end{definition}

\noindent
We now show, by induction on the complexity of types, that 
\begin{subclaim*}
Every profile  $\synseq q q {w_x}$, with $w_x$ a type having complexity $n$, can be derived by a proof of size $\leq g(n)$.
\end{subclaim*}
\begin{proof}[Proof of the subclaim] \mbox{ }

\smallskip
\noindent
\textbf{\underline{Base case:} how to derive $\synseq q q {w_x}$ with $w_x$ a type of complexity $0$.}

The fact that $w_x$ has complexity $0$ means that $\rho_h^x$ has no inner vertex labeled by $q$. So by induction hypothesis, there exists a derivation of $\synsneq q q {w_x}$ 
of size $f(q-1, |w_x|)$.

Since $|w_x| \leq |w \cup P| \leq |w \cup Q |$, where $Q$ is the set of states in the zero automaton, we get that the size of the proof of  $\synseq q q {w_x}$ has size $\leq g(0)= f(q-1, |w|+|Q|)$, as desired.

\noindent
\textbf{\underline{Inductive case:} how to derive $\synseq q q {w_x}$ with $w_x$ a type of complexity $n+1$.}

Let $\rho^x_h|_q$ be the partial run $\rho_h^x$ where all vertices below the first occurrences of $q$'s have been removed. 
The partial run
$\rho^x_h|_q$ has profile $\synsneq q q {w^\prime_x}\cup \{q, \dots, q\}$ (note the strict inequality $<q$), for some $w^\prime_x\subseteq w_x$ (this is because some ports in $w_x$ might not appear in $\rho^x_h|_q$ since some vertices where removed) and  where $\{ q,\dots, q\}$ denotes the multiset of leaves labeled by $q$ in $\rho^x_h|_q$  (up to a maximal multiplicity of $|w_x|$). 
Note that \[|{w^\prime_x}\cup \{q, \dots, q\}| \leq 2 |w_x| \leq 2(|w|+|Q|).\]

From the the inductive assumption on $q$, we know that the profile $\synsneq q q {w^\prime_x}\cup \{q, \dots, q\}$ can be proved using a derivation of  size at most $f(q-1, | {w^\prime_x}\cup \{q, \dots, q\}|)$. Moreover, we have 
\[f(q-1, | {w^\prime_x}\cup \{q, \dots, q\}|) \leq f(q-1, 2(|w|+|Q|)),\]
because $f$ is non-decreasing in the second coordinate. 

Now that we have constructed a derivation of $\synsneq q q {w^\prime_x}\cup \{q, \dots, q\}$ we can combine it with the derivations of $\synseq q q {w_y}$ corresponding to all vertices $y$, labeled by $q$, which appears as leaves in $\rho^x_h|_q$. These derivations can be constructed by induction hypothesis on the complexity of the types $w_y$ (the fact that each $w_y$ has complexity $\leq n$ follows, by definition, from the assumption that $w_x$ has type $n+1$). So, by means of application of the (U) and (D) rule (see Figure \ref{proof_u_d}) we can construct the desired derivation of $\synseq q q {w_x}$.
This completes the proof of the subclaim.
\end{proof}


\noindent
Since the derivation of $\synseq q q {w_x}$ is obtained by combining:
\begin{enumerate}
\item the derivation of $\synsneq q q {w^\prime_x}\cup \{q, \dots, q\}$, having size  $\leq f(q-1, 2(|w|+|Q|))$  
\item the $|w_x|$-many derivations $\synseq q q {w_y}$, having size $\leq g(n)$, 
\end{enumerate}
and recalling that $|w_x|\leq |w|+|Q|$, we have that the size of the proof $\synseq q q {w_x}$ is smaller or equal than:
$$f(q-1, 2(|w|+|Q|)) + g(n)\cdot (|w|+|Q|) +  (|w|+|Q|)^2 = g(n+1) $$
where the expression $(|w|+|Q|)^2$ counts the number of applications of the (U) and (D) rules.

We have established that, for every type $w_x$ of complexity $n$, the profile $\synseq q q {w_x}$ can be derived by a proof of size $\leq g(n)$. Since the number of types $v\subseteq w \cup P$ is bounded by $2^{|w| + |Q|}$,
each type has complexity at most $2^{|w| + |Q|}$. Hence, we know that an arbitrary profile $\synseq q q {w_x}$ can be derived by a derivation of size $\leq g(2^{|w|+|Q|})$. \\

Now we are ready to conclude the proof of Claim \ref{claim:frontier_estimation}. We need to construct a proof of the profile   $\synseq p q {w\cup P}$ of $\rho_h$ of size at most $N$.
Let $\rho_h|_q$ denote the profile $\rho_h$ where all vertices below the first occurrences of $q$ have been removed. Then $\rho_h|_q$ has profile  $\synsneq p q w^\prime \cup \{q,\dots, q\}$ (note the strict inequality) where $w^\prime\subseteq  w$ and, as before $\{ q,\dots, q\}$, denotes the multiset of leaves labeled by $q$ in $\rho_h|_q$  (up to a maximal multiplicity of $|w|$).  Since $|w^\prime|\leq |w|$ and $|\{q,\dots, q\}|\leq |w|$, a derivation of $\synsneq p q w^\prime \cup \{q,\dots, q\}$ of size $\leq f(q-1, 2|w|)$ can be obtained by induction hypothesis on $q$. 

Using this derivation and the appropriate required derivations $\synseq q q w_x$ (each having size bounded by $g(2^{|w|+|Q|})$), we can get a derivation of size 
\begin{equation}
N \leq f(q-1, 2|w|) + g(2^{|w|+|Q|})\cdot {|w|} + |w|^2,
\label{equation_bound_1}
\end{equation}
where the expression $|w|^2$ counts the number of $(U)$ and $(D)$ rules required to combine the sub-derivations. 
This finishes the proof of Claim \ref{claim:frontier_estimation}.
\end{proof}

 \paragraph{How to construct a derivation of: $\synseq {p_i} q {w_i}$.\\}

In this section we describe how to construct a derivation of $\synseq {p_i} q {w_i}$ for each $p_i\!\in\! P$. 
We consider separately the three cases: (1) $q\not\in\qall$ (with either $q\!\in\!\qzero$ or  $q\!\not\in\!\qzero$), (2) $q\in\qall$ and $q\!\not\in\!\qzero$, and (3) $q\in\qall$ and $q\!\in\!\qzero$.

\paragraph*{Subcase $q\not\in\qall$.}

Consider the sub-partial run $\rho_i$ rooted at one leaf labeled by $p_i$ having profile $\synseq {p_i} q {w_i}$, for some $w_i\subseteq w$. For example, in the image depicted, the sub-partial run rooted at $p_1$.

\def\prs{\tikz[scale=.65, every node/.style={scale=0.65}, baseline=1ex,shorten >=.1pt,node distance=1.8cm,on grid,semithick,auto,
every state/.style={fill=white,draw=black,circular drop shadow,inner sep=0mm,text=black},
accepting/.style ={fill=gray,text=white}]{
\node[state,accepting] (p) {$p_i$};
\begin{scope}[scale=.55, every node/.style={scale=0.55}, baseline=1ex,shorten >=.1pt,node distance=1.8cm, semithick,auto,
every state/.style={fill=white,draw=black,circular drop shadow,inner sep=0mm,text=black},
accepting/.style ={fill=gray,text=white}]
\node[draw=none] (r) [below left=1.7cm and 1.4cm  of p] {};
\node[draw=none] (s) [below right=1.7cm and 1.4cm  of p] {};
\end{scope}
\draw [black] (p) -- (r);
\draw [black] (p) -- (s); 
\node[draw=none] (q2m) [below left=1.4cm and 0.0cm of p, draw = none] {\ldots};
}
}
\def\sqsone{
\begin{scope}[scale=.45, every node/.style={scale=0.45}, baseline=1ex,shorten >=.1pt,node distance=1.8cm, semithick,auto,
every state/.style={fill=white,draw=black,circular drop shadow,inner sep=0mm,text=black},
accepting/.style ={fill=gray,text=white}]
\node[state] (q2r) [below left=1.5cm and 0.7cm of s] {$q_1$};
\node[state] (s2) [below right=1.5cm and 0.7cm of s] {$q_2$};
\draw [black] (s) -- (q2r);
\draw [black] (s) -- (s2); 
\end{scope}
}
\def\rqp{
\begin{scope}[scale=.45, every node/.style={scale=0.45}, baseline=1ex,shorten >=.1pt,node distance=1.8cm, semithick,auto,
every state/.style={fill=white,draw=black,circular drop shadow,inner sep=0mm,text=black},
accepting/.style ={fill=gray,text=white}]
\node[state] (q2) [below left=1.5cm and 0.7cm of r] {$q_1$};
\node[state] (p2) [below right=1.5cm and 0.7cm of r] {$p_2$};
\draw [black] (r) -- (q2);
\draw [black] (r) -- (p2); 
\end{scope}
}
\def\sqstwo{
\begin{scope}[scale=.35, every node/.style={scale=0.35}, baseline=1ex,shorten >=.1pt,node distance=1.8cm, semithick,auto,
every state/.style={fill=white,draw=black,circular drop shadow,inner sep=0mm,text=black},
accepting/.style ={fill=gray,text=white}]
\node[state] (q3) [below left=1.35cm and 0.4cm of s2] {$q_3$};
\node[state] (p3) [below right=1.35cm and 0.4cm of s2] {$p_1$};
\draw [black] (s2) -- (q3);
\draw [black] (s2) -- (p3); 
\end{scope}
}
\def\sqsthree{
\begin{scope}[scale=.35, every node/.style={scale=0.35}, baseline=1ex,shorten >=.1pt,node distance=1.8cm, semithick,auto,
every state/.style={fill=white,draw=black,circular drop shadow,inner sep=0mm,text=black},
accepting/.style ={fill=gray,text=white}]
\node[state] (q3) [below left=1.35cm and 0.4cm of q2r] {$q_2$};
\node[state] (s3p) [below right=1.35cm and 0.4cm of q2r] {$p_2$};
\draw [black] (q2r) -- (q3);
\draw [black] (q2r) -- (s3p); 
\end{scope}
}
\def\sqsfour{
\begin{scope}[scale=.35, every node/.style={scale=0.35}, baseline=1ex,shorten >=.1pt,node distance=1.8cm, semithick,auto,
every state/.style={fill=white,draw=black,circular drop shadow,inner sep=0mm,text=black},
accepting/.style ={fill=gray,text=white}]
\node[state] (q3) [below left=1.35cm and 0.4cm of p2] {$q_3$};
\node[state] (s3) [below right=1.35cm and 0.4cm of p2] {$p_1$};
\draw [black] (p2) -- (q3);
\draw [black] (p2) -- (s3); 
\end{scope}
}
\def\sqsfive{
\begin{scope}[scale=.35, every node/.style={scale=0.35}, baseline=1ex,shorten >=.1pt,node distance=1.8cm, semithick,auto,
every state/.style={fill=white,draw=black,circular drop shadow,inner sep=0mm,text=black},
accepting/.style ={fill=gray,text=white}]
\node[state] (q3) [below left=1.35cm and 0.4cm of q2] {$q_1$};
\node[state] (s3) [below right=1.35cm and 0.4cm of q2] {$p$};
\draw [black] (q2) -- (q3);
\draw [black] (q2) -- (s3); 
\end{scope}
}
\def\sqssix{
\begin{scope}[scale=.15, every node/.style={scale=0.15}, baseline=1ex,shorten >=.1pt,node distance=1.8cm, semithick,auto,
every state/.style={fill=white,draw=black,circular drop shadow,inner sep=0mm,text=black},
accepting/.style ={fill=gray,text=white}]
\node[state] (q5) [below left=1cm and 0.4cm of p3] {$p$};
\node[state] (s5) [below right=1cm and 0.4cm of p3] {$p$};
\draw [black] (p3) -- (q5);
\draw [black] (p3) -- (s5); 
\end{scope}
\node[draw=none] (s6) [below left=0.4cm and 0.35cm of s5,draw = none] {\ldots};
}
\def\sqsseven{
\begin{scope}[scale=.08, every node/.style={scale=0.08}, baseline=1ex,shorten >=.1pt,node distance=1.8cm, semithick,auto,
every state/.style={fill=white,draw=black,circular drop shadow,inner sep=0mm,text=black},
accepting/.style ={fill=gray,text=white}]
\node[draw=none] (q6) [below left=0.7cm and 0.3cm of s3] {};
\draw [black] (s3) -- (q6);
\node[draw=none] (s6) [below right=0.7cm and 0.3cm of s3,draw = none] {};
\draw [black] (s3) -- (s6); 
\end{scope}
\node[draw=none] (q6m) [below left=0.7cm and 0.0cm of s3, draw = none] {\ldots};
}
\def\sqseight{
\begin{scope}[scale=.08, every node/.style={scale=0.08}, baseline=1ex,shorten >=.1pt,node distance=1.8cm, semithick,auto,
every state/.style={fill=white,draw=black,circular drop shadow,inner sep=0mm,text=black},
accepting/.style ={fill=gray,text=white}]
\node[draw=none] (q8) [below left=0.7cm and 0.3cm of s3p] {};
\draw [black] (s3p) -- (q8);
\node[draw=none] (s8) [below right=0.7cm and 0.3cm of s3p,draw = none] {};
\draw [black] (s3p) -- (s8); 
\end{scope}
\node[draw=none] (q8m) [below left=0.7cm and 0.0cm of s3p, draw = none] {\ldots};
}
\begin{figure}[H]
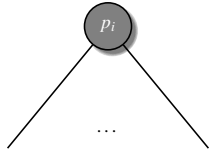

\centering
\prs
\caption{Sub-partial run rooted at $p_i$.}
\label{fig:matteo2}
\end{figure}

Our goal now is to derive the profile $\synseq {p_i} q {w_i}$.
Since $q$ is the maximal inner state in this sub-partial run, and it is not in $\Qall$, we conclude that there are no infinite paths in $\rho_i$ having infinitely many occurrences of $q$. 

If there are no vertices labeled by $q$ at all in $\rho_i$, then the profile of $\rho_i$ is actually $\synsneq {p_i} {q} {w_i}$ and this is derivable by inductive hypothesis on $q$ by a proof of size bounded by $f(q-1, |w_i|)$.

So assume there exist some vertices in $\rho_i$ labeled by $q$. For any such vertex $x$, let us denote with $\rho^x_i$ the sub-partial run of $\rho_i$ rooted at $x$. The partial run $\rho^x_i$ has profile $\synsneq {q} {q} {w_i^x}$, for some $w_i^x\subseteq w_i$. The multiset $w_i^x$ is called \emph{the type of $x$}. 

Following the same idea presented earlier, we prove how to derive all these profiles $\synsneq {q} {q} {w_i^x}$ by induction on the appropriate notion  of complexity of types. The notion of type is similar but technically different than the one introduced in the proof of Claim \ref{claim:frontier_estimation}. 

\begin{definition}
A type $v\subseteq w_i$ is of \emph{complexity $0$} if there exists some $x$ (labeled by $q$) in $\rho_i$ such that:
\begin{enumerate}
\item there are no vertices labeled by $q$ below $x$ (excluded $x$ itself) 
\item $x$ has type $v$, that is $v=w_i^x$.
\end{enumerate}
A type $v\subseteq w_i$ has complexity $n+1$ if there exists some $x$ (labeled by $q$) in $\rho_i$ such that:
\begin{enumerate}
\item the first (if any)  state $y$ labeled by $q$ in any path from $x$ (excluded) has type of complexity $\leq n$,
\item $x$ has type $v$, that is $v=w_i^x$.
\end{enumerate}
\end{definition}

 From the fact that there are no infinite paths with infinitely many $q$'s in $\rho_i$, we deduce that each type has a finite complexity. Furthermore, there are at most $2^{|w_i|}\leq 2^{|w|}$ types (i.e., multisubsets of $w_i$). 
\begin{claim}
Each type $w_i^x$ of complexity $n$, the profile $\synsneq {q} {q} {w_i^x}$ is derivable by a proof of size smaller or equal than $h(n)$ where the function $h\!:\!\mathbb{N}\rightarrow\mathbb{N}$ is defined as in Definition \ref{definition_f} as:
\begin{itemize}
\item $h(0) =  f(q-1, |w|)$,
\item $h(n+1) = f(q-1, 2|w|) + (h(n) \cdot |w|) + |w|^2$.
\end{itemize}
\label{estimation_qnotq_all}
\end{claim}

\noindent
From the claim follows that each $\synsneq {q} {q} {w_i^x}$ can be derived by a proof of size $\leq h(2^{|w|})$.

\begin{proof}[Proof of the Claim.]
$\ $\\ 
\textbf{\underline{Base case:} how to derive $\synseq q q {w_i^x}$ given $w_i^x$ a type of complexity $0$.}

The fact that $w_i^x$ has complexity $0$ means that $\rho_i^x$ has no inner vertices (excluded the root $x$ itself) labeled by $q$. So by induction hypothesis, there exists a derivation of $\synsneq q q {w_i^x}$ 
of size $f(q-1, |w_i^x|)$. Since $|w_i^x| \leq | w_i | \leq |w| $, the proof has size $\leq f(q-1, |w|) = h(0)$, as desired.

\noindent
\textbf{\underline{Inductive case:} how to derive $\synseq q q {w_i^x}$ given $w_i^x$, a type of complexity $n+1$.}

Let $\rho^x_i |_q$ be the partial run $\rho^x_i$ where all vertices below the first occurrences of $q$'s have been removed. 

\def\prs{\tikz[scale=.65, every node/.style={scale=0.65}, baseline=1ex,shorten >=.1pt,node distance=1.8cm,on grid,semithick,auto,
every state/.style={fill=white,draw=black,circular drop shadow,inner sep=0mm,text=black},
accepting/.style ={fill=gray,text=white}]{
\node[state,accepting] (p) {$q$};
\begin{scope}[scale=.55, every node/.style={scale=0.55}, baseline=1ex,shorten >=.1pt,node distance=1.8cm, semithick,auto,
every state/.style={fill=white,draw=black,circular drop shadow,inner sep=0mm,text=black},
accepting/.style ={fill=gray,text=white}]
\node[draw=none] (r) [below left=1.7cm and 1.4cm  of p] {};
\node[draw=none] (s) [below right=1.7cm and 1.4cm  of p] {};
\end{scope}
\draw [black] (p) -- (r);
\draw [black] (p) -- (s); 
\node[draw=none] (q) [below left=1.4cm and 0.0cm of p, draw = none] {$q$};
\draw [black,decorate,decoration=snake] (p) -- (q); 
}
}
\begin{figure}[H]
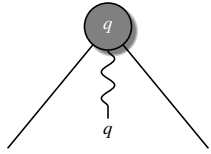

\centering
\prs
\caption{Removing all vertices below $q$'s in the partial run $\rho^{x}_i$.}
\label{fig:matteo4}
\end{figure}

The partial run
$\rho^x_i|_q$ has profile $\synsneq q q {w^\prime_x}\cup \{q, \dots, q\}$ (note the strict inequality $<q$), for some $w^\prime_x\subseteq w_i^x$ (this is because some ports in $w_x^i$ might not appear in $\rho^x_i|_q$ since some vertices were removed) and  where $\{ q,\dots, q\}$ denotes the multiset of leaves labeled by $q$ in $\rho^x_i|_q$  (up to a maximal multiplicity of $|w_i^x|$). 
Note that \[|{w^\prime_x}\cup \{q, \dots, q\}| \leq  2 |w_i^x| \leq 2|w|.\]

From the inductive assumption on $q$ the profile $\synsneq q q {w^\prime_x}\cup \{q, \dots, q\}$ can be proved using a derivation of  size at most $f(q-1, | {w^\prime_x}\cup \{q, \dots, q\}|)$. Moreover, we have 
\[f(q-1, | {w^\prime_x}\cup \{q, \dots, q\}|) \leq f(q-1, 2|w|),\]
because $f$ is non-decreasing in the second coordinate. 

Now we are ready to construct the desired derivation $\synseq q q {w_i^x}$.
We construct it using the derivation of $\synsneq q q {w^\prime_x}\cup \{q, \dots, q\}$, rules $(U)$ and $(D)$ and appropriate auxiliary derivations $\synsneq q q {w_i^y}$ of type $\leq n$.
The resulting proof has then size $ \leq f(q-1, 2|w|) + (h(n) \cdot |w|) + |w|^2= h(n+1)$.
\end{proof}

We are ready to  construct a proof of the profile   $\synseq {p_i} q {w_i}$ of $\rho_i$ and
prove an appropriate estimation of its size (see Equation \ref{equation_bound_2}).  
Let $\rho_i|_q$ denote the profile $\rho_i$ where all vertices below the first occurrences of $q$ have been removed. Then $\rho_i|_q$ has profile  $\synsneq p q w^\prime \cup \{q,\dots, q\}$ (note the strict inequality) where $w^\prime\subseteq  w_i$ and, as before $\{ q,\dots, q\}$, denotes the multiset of leaves labeled by $q$ in $\rho_i|_q$  (up to a maximal multiplicity of $|w|$).  Since $|w^\prime| \leq |w_i| \leq |w|$ and $|\{q,\dots, q\}|\leq |w|$, a derivation of $\synsneq p q w^\prime \cup \{q,\dots, q\}$ of size $\leq f(q-1, 2|w|)$ can be obtained by induction hypothesis on $q$. 

Using this derivation and the appropriate required derivations $\synseq q q w_x$ (having size bounded by $h(2^{|w|})$), we can get a derivation of size 

\begin{equation}
M_i \leq f(q-1, 2|w|) + (h(2^{|w|})\cdot |w|) + |w|^2.
\label{equation_bound_2}
\end{equation}

\paragraph*{Subcase $q\in\qall$ and $q\!\not\in \qzero$.}
Consider the sub-partial run $\rho_i$ rooted at  $p_i$. 
\def\prs{\tikz[scale=.65, every node/.style={scale=0.65}, baseline=1ex,shorten >=.1pt,node distance=1.8cm,on grid,semithick,auto,
every state/.style={fill=white,draw=black,circular drop shadow,inner sep=0mm,text=black},
accepting/.style ={fill=gray,text=white}]{
\node[state,accepting] (p) {$p_i$};
\begin{scope}[scale=.55, every node/.style={scale=0.55}, baseline=1ex,shorten >=.1pt,node distance=1.8cm, semithick,auto,
every state/.style={fill=white,draw=black,circular drop shadow,inner sep=0mm,text=black},
accepting/.style ={fill=gray,text=white}]
\node[draw=none] (r) [below left=1.7cm and 1.4cm  of p] {};
\node[draw=none] (s) [below right=1.7cm and 1.4cm  of p] {};
\end{scope}
\draw [black] (p) -- (r);
\draw [black] (p) -- (s); 
\node[draw=none] (q2m) [below left=1.4cm and 0.0cm of p, draw = none] {\ldots};
}
}
\def\sqsone{
\begin{scope}[scale=.45, every node/.style={scale=0.45}, baseline=1ex,shorten >=.1pt,node distance=1.8cm, semithick,auto,
every state/.style={fill=white,draw=black,circular drop shadow,inner sep=0mm,text=black},
accepting/.style ={fill=gray,text=white}]
\node[state] (q2r) [below left=1.5cm and 0.7cm of s] {$q_1$};
\node[state] (s2) [below right=1.5cm and 0.7cm of s] {$q_2$};
\draw [black] (s) -- (q2r);
\draw [black] (s) -- (s2); 
\end{scope}
}
\def\rqp{
\begin{scope}[scale=.45, every node/.style={scale=0.45}, baseline=1ex,shorten >=.1pt,node distance=1.8cm, semithick,auto,
every state/.style={fill=white,draw=black,circular drop shadow,inner sep=0mm,text=black},
accepting/.style ={fill=gray,text=white}]
\node[state] (q2) [below left=1.5cm and 0.7cm of r] {$q_1$};
\node[state] (p2) [below right=1.5cm and 0.7cm of r] {$p_2$};
\draw [black] (r) -- (q2);
\draw [black] (r) -- (p2); 
\end{scope}
}
\def\sqstwo{
\begin{scope}[scale=.35, every node/.style={scale=0.35}, baseline=1ex,shorten >=.1pt,node distance=1.8cm, semithick,auto,
every state/.style={fill=white,draw=black,circular drop shadow,inner sep=0mm,text=black},
accepting/.style ={fill=gray,text=white}]
\node[state] (q3) [below left=1.35cm and 0.4cm of s2] {$q_3$};
\node[state] (p3) [below right=1.35cm and 0.4cm of s2] {$p_1$};
\draw [black] (s2) -- (q3);
\draw [black] (s2) -- (p3); 
\end{scope}
}
\def\sqsthree{
\begin{scope}[scale=.35, every node/.style={scale=0.35}, baseline=1ex,shorten >=.1pt,node distance=1.8cm, semithick,auto,
every state/.style={fill=white,draw=black,circular drop shadow,inner sep=0mm,text=black},
accepting/.style ={fill=gray,text=white}]
\node[state] (q3) [below left=1.35cm and 0.4cm of q2r] {$q_2$};
\node[state] (s3p) [below right=1.35cm and 0.4cm of q2r] {$p_2$};
\draw [black] (q2r) -- (q3);
\draw [black] (q2r) -- (s3p); 
\end{scope}
}
\def\sqsfour{
\begin{scope}[scale=.35, every node/.style={scale=0.35}, baseline=1ex,shorten >=.1pt,node distance=1.8cm, semithick,auto,
every state/.style={fill=white,draw=black,circular drop shadow,inner sep=0mm,text=black},
accepting/.style ={fill=gray,text=white}]
\node[state] (q3) [below left=1.35cm and 0.4cm of p2] {$q_3$};
\node[state] (s3) [below right=1.35cm and 0.4cm of p2] {$p_1$};
\draw [black] (p2) -- (q3);
\draw [black] (p2) -- (s3); 
\end{scope}
}
\def\sqsfive{
\begin{scope}[scale=.35, every node/.style={scale=0.35}, baseline=1ex,shorten >=.1pt,node distance=1.8cm, semithick,auto,
every state/.style={fill=white,draw=black,circular drop shadow,inner sep=0mm,text=black},
accepting/.style ={fill=gray,text=white}]
\node[state] (q3) [below left=1.35cm and 0.4cm of q2] {$q_1$};
\node[state] (s3) [below right=1.35cm and 0.4cm of q2] {$p$};
\draw [black] (q2) -- (q3);
\draw [black] (q2) -- (s3); 
\end{scope}
}
\def\sqssix{
\begin{scope}[scale=.15, every node/.style={scale=0.15}, baseline=1ex,shorten >=.1pt,node distance=1.8cm, semithick,auto,
every state/.style={fill=white,draw=black,circular drop shadow,inner sep=0mm,text=black},
accepting/.style ={fill=gray,text=white}]
\node[state] (q5) [below left=1cm and 0.4cm of p3] {$p$};
\node[state] (s5) [below right=1cm and 0.4cm of p3] {$p$};
\draw [black] (p3) -- (q5);
\draw [black] (p3) -- (s5); 
\end{scope}
\node[draw=none] (s6) [below left=0.4cm and 0.35cm of s5,draw = none] {\ldots};
}
\def\sqsseven{
\begin{scope}[scale=.08, every node/.style={scale=0.08}, baseline=1ex,shorten >=.1pt,node distance=1.8cm, semithick,auto,
every state/.style={fill=white,draw=black,circular drop shadow,inner sep=0mm,text=black},
accepting/.style ={fill=gray,text=white}]
\node[draw=none] (q6) [below left=0.7cm and 0.3cm of s3] {};
\draw [black] (s3) -- (q6);
\node[draw=none] (s6) [below right=0.7cm and 0.3cm of s3,draw = none] {};
\draw [black] (s3) -- (s6); 
\end{scope}
\node[draw=none] (q6m) [below left=0.7cm and 0.0cm of s3, draw = none] {\ldots};
}
\def\sqseight{
\begin{scope}[scale=.08, every node/.style={scale=0.08}, baseline=1ex,shorten >=.1pt,node distance=1.8cm, semithick,auto,
every state/.style={fill=white,draw=black,circular drop shadow,inner sep=0mm,text=black},
accepting/.style ={fill=gray,text=white}]
\node[draw=none] (q8) [below left=0.7cm and 0.3cm of s3p] {};
\draw [black] (s3p) -- (q8);
\node[draw=none] (s8) [below right=0.7cm and 0.3cm of s3p,draw = none] {};
\draw [black] (s3p) -- (s8); 
\end{scope}
\node[draw=none] (q8m) [below left=0.7cm and 0.0cm of s3p, draw = none] {\ldots};
}
\begin{figure}[H]
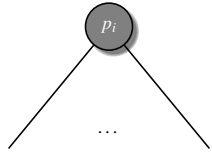

\centering
\prs
\caption{Sub-partial run rooted at $p_i$.}
\label{fig:matteo2_bis}
\end{figure}


Our goal now is to derive the profile $\synseq {p_i} q {w_i}$. If there are not states labeled by $q$ in $\rho_i$, then the partial run $\rho_i$ actually has profile $\synsneq {p_i} q {w_i}$ and this is derivable by inductive hypothesis on $q$ with a proof of size $\leq f(q-1, |w_i|) \leq f(q-1, |w|)$. 
So let us assume that there are some states labeled by $q$ in $\rho_i$.

Note, that unlike the previous case, this time it is possible that the partial run $\rho_i$ contains infinite paths with infinitely many $q$ as this would satisfy the $\qall$ condition and would not constitute a problem with respect to the $\qzero$ conditions, since $q\!\not\in \qzero$.
Let $\rho_i|_q$ be obtained by $\rho_i$ by removing all vertices below the first occurrences of states labeled by $q$.

\def\prs{\tikz[scale=.65, every node/.style={scale=0.65}, baseline=1ex,shorten >=.1pt,node distance=1.8cm,on grid,semithick,auto,
every state/.style={fill=white,draw=black,circular drop shadow,inner sep=0mm,text=black},
accepting/.style ={fill=gray,text=white}]{
\node[state,accepting] (p) {$p_i$};
\begin{scope}[scale=.55, every node/.style={scale=0.55}, baseline=1ex,shorten >=.1pt,node distance=1.8cm, semithick,auto,
every state/.style={fill=white,draw=black,circular drop shadow,inner sep=0mm,text=black},
accepting/.style ={fill=gray,text=white}]
\node[draw=none] (r) [below left=1.7cm and 1.4cm  of p] {};
\node[draw=none] (s) [below right=1.7cm and 1.4cm  of p] {};
\end{scope}
\draw [black] (p) -- (r);
\draw [black] (p) -- (s); 
\node[draw=none] (q) [below left=1.4cm and 0.0cm of p, draw = none] {$q$};
\draw [black,decorate,decoration=snake] (p) -- (q); 
}
}
\begin{figure}[H]
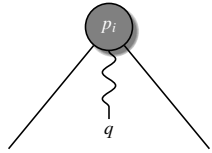

\centering
\prs
\caption{Removing all vertices below $q$'s in the partial run $\rho_i$.}
\label{fig:matteo4bis}
\end{figure}

 Note, that the obtained tree $\rho_i|_q$ is itself a partial run with profile $\synsneq {p_i} {q} {w^\prime \cup \{ q, \dots, q\}}$, for a multiset $w^\prime$ contained in $w_i$ and where $\{ q,\dots, q\}$ denotes the multiset of leaves labeled by $q$  up to a maximal multiplicity of $|w_i|$. Moreover, in $\rho_i|_q$ all inner nodes of $\rho^\prime$ are different than $q$. Hence we can derive $\synsneq {p_i}{q} {w^\prime\cup \{ q,\dots, q\}}$ (note the strict inequality) by induction hypothesis on $q$ with a proof of size $\leq f(q-1, 2|w|)$.

Now, if $p_i\!=\!q$, we can obtain the desired derivation of $\synseq {p_i} q {w^\prime}$ by application of the Weak Looping (WL) rule:

\begin{prooftree}
\AxiomC{$\synsneq {q} {q} {w^\prime \cup  \{q,\dots, q\}}$}
\LeftLabel{D}
\UnaryInfC{$\synsneq {q} {q} {w^\prime \cup  \{q\}}$}
\LeftLabel{WL}
\UnaryInfC{$\synsneq {q} {q} {w^\prime}$} 
\end{prooftree}

If instead $p_i\neq q$, then for all leaves in $\rho_i |_q$ labeled by $q$, we can derive  $\synsneq {q} {q} {w^{\prime\prime}}$, for some $w^{\prime\prime}\!\subseteq\!w_i$ as described just above for the case $p_i\!=\!q$. Note, that there are at most $2^{|w|}$ such derivations, as $w^{\prime\prime}\subseteq w_i \subseteq w$, and each derivation has size bounded by $f(q-1, 2|w|)$.

We can then combine the proof of the profile $\synsneq {p_i} {q} {w^\prime \cup \{ q,\dots, q\}}$ (having size $\leq f(q-1,2|w|)$) with all the relevant profiles $\synsneq {q} {q} {w^{\prime\prime}}$ (having size $ \leq f(q-1, 2|w|)$ ) to obtain the desired derivation. The final proof has size:

\begin{equation}
M_i  \leq f(q-1, 2|w|) + (f(q-1, 2|w|)\cdot 2^{|w|} ) + |w|^2 =  f(q-1, 2|w|)\cdot (2^{|w|} + 1) + |w|^2,
\label{equation_bound_3}
\end{equation}
where the last term $|w|^2$ counts the number of applications of $(U)$ and $(D)$ rule to combine all sub-derivations into the final derivation.




\paragraph*{Subcase $q\in\qall$ and $q\!\in \qzero$.}
Consider the sub-partial run $\rho_i$ rooted at  $p_i$. 
\def\prs{\tikz[scale=.65, every node/.style={scale=0.65}, baseline=1ex,shorten >=.1pt,node distance=1.8cm,on grid,semithick,auto,
every state/.style={fill=white,draw=black,circular drop shadow,inner sep=0mm,text=black},
accepting/.style ={fill=gray,text=white}]{
\node[state,accepting] (p) {$p_i$};
\begin{scope}[scale=.55, every node/.style={scale=0.55}, baseline=1ex,shorten >=.1pt,node distance=1.8cm, semithick,auto,
every state/.style={fill=white,draw=black,circular drop shadow,inner sep=0mm,text=black},
accepting/.style ={fill=gray,text=white}]
\node[draw=none] (r) [below left=1.7cm and 1.4cm  of p] {};
\node[draw=none] (s) [below right=1.7cm and 1.4cm  of p] {};
\end{scope}
\draw [black] (p) -- (r);
\draw [black] (p) -- (s); 
\node[draw=none] (q2m) [below left=1.4cm and 0.0cm of p, draw = none] {\ldots};
}
}
\def\sqsone{
\begin{scope}[scale=.45, every node/.style={scale=0.45}, baseline=1ex,shorten >=.1pt,node distance=1.8cm, semithick,auto,
every state/.style={fill=white,draw=black,circular drop shadow,inner sep=0mm,text=black},
accepting/.style ={fill=gray,text=white}]
\node[state] (q2r) [below left=1.5cm and 0.7cm of s] {$q_1$};
\node[state] (s2) [below right=1.5cm and 0.7cm of s] {$q_2$};
\draw [black] (s) -- (q2r);
\draw [black] (s) -- (s2); 
\end{scope}
}
\def\rqp{
\begin{scope}[scale=.45, every node/.style={scale=0.45}, baseline=1ex,shorten >=.1pt,node distance=1.8cm, semithick,auto,
every state/.style={fill=white,draw=black,circular drop shadow,inner sep=0mm,text=black},
accepting/.style ={fill=gray,text=white}]
\node[state] (q2) [below left=1.5cm and 0.7cm of r] {$q_1$};
\node[state] (p2) [below right=1.5cm and 0.7cm of r] {$p_2$};
\draw [black] (r) -- (q2);
\draw [black] (r) -- (p2); 
\end{scope}
}
\def\sqstwo{
\begin{scope}[scale=.35, every node/.style={scale=0.35}, baseline=1ex,shorten >=.1pt,node distance=1.8cm, semithick,auto,
every state/.style={fill=white,draw=black,circular drop shadow,inner sep=0mm,text=black},
accepting/.style ={fill=gray,text=white}]
\node[state] (q3) [below left=1.35cm and 0.4cm of s2] {$q_3$};
\node[state] (p3) [below right=1.35cm and 0.4cm of s2] {$p_1$};
\draw [black] (s2) -- (q3);
\draw [black] (s2) -- (p3); 
\end{scope}
}
\def\sqsthree{
\begin{scope}[scale=.35, every node/.style={scale=0.35}, baseline=1ex,shorten >=.1pt,node distance=1.8cm, semithick,auto,
every state/.style={fill=white,draw=black,circular drop shadow,inner sep=0mm,text=black},
accepting/.style ={fill=gray,text=white}]
\node[state] (q3) [below left=1.35cm and 0.4cm of q2r] {$q_2$};
\node[state] (s3p) [below right=1.35cm and 0.4cm of q2r] {$p_2$};
\draw [black] (q2r) -- (q3);
\draw [black] (q2r) -- (s3p); 
\end{scope}
}
\def\sqsfour{
\begin{scope}[scale=.35, every node/.style={scale=0.35}, baseline=1ex,shorten >=.1pt,node distance=1.8cm, semithick,auto,
every state/.style={fill=white,draw=black,circular drop shadow,inner sep=0mm,text=black},
accepting/.style ={fill=gray,text=white}]
\node[state] (q3) [below left=1.35cm and 0.4cm of p2] {$q_3$};
\node[state] (s3) [below right=1.35cm and 0.4cm of p2] {$p_1$};
\draw [black] (p2) -- (q3);
\draw [black] (p2) -- (s3); 
\end{scope}
}
\def\sqsfive{
\begin{scope}[scale=.35, every node/.style={scale=0.35}, baseline=1ex,shorten >=.1pt,node distance=1.8cm, semithick,auto,
every state/.style={fill=white,draw=black,circular drop shadow,inner sep=0mm,text=black},
accepting/.style ={fill=gray,text=white}]
\node[state] (q3) [below left=1.35cm and 0.4cm of q2] {$q_1$};
\node[state] (s3) [below right=1.35cm and 0.4cm of q2] {$p$};
\draw [black] (q2) -- (q3);
\draw [black] (q2) -- (s3); 
\end{scope}
}
\def\sqssix{
\begin{scope}[scale=.15, every node/.style={scale=0.15}, baseline=1ex,shorten >=.1pt,node distance=1.8cm, semithick,auto,
every state/.style={fill=white,draw=black,circular drop shadow,inner sep=0mm,text=black},
accepting/.style ={fill=gray,text=white}]
\node[state] (q5) [below left=1cm and 0.4cm of p3] {$p$};
\node[state] (s5) [below right=1cm and 0.4cm of p3] {$p$};
\draw [black] (p3) -- (q5);
\draw [black] (p3) -- (s5); 
\end{scope}
\node[draw=none] (s6) [below left=0.4cm and 0.35cm of s5,draw = none] {\ldots};
}
\def\sqsseven{
\begin{scope}[scale=.08, every node/.style={scale=0.08}, baseline=1ex,shorten >=.1pt,node distance=1.8cm, semithick,auto,
every state/.style={fill=white,draw=black,circular drop shadow,inner sep=0mm,text=black},
accepting/.style ={fill=gray,text=white}]
\node[draw=none] (q6) [below left=0.7cm and 0.3cm of s3] {};
\draw [black] (s3) -- (q6);
\node[draw=none] (s6) [below right=0.7cm and 0.3cm of s3,draw = none] {};
\draw [black] (s3) -- (s6); 
\end{scope}
\node[draw=none] (q6m) [below left=0.7cm and 0.0cm of s3, draw = none] {\ldots};
}
\def\sqseight{
\begin{scope}[scale=.08, every node/.style={scale=0.08}, baseline=1ex,shorten >=.1pt,node distance=1.8cm, semithick,auto,
every state/.style={fill=white,draw=black,circular drop shadow,inner sep=0mm,text=black},
accepting/.style ={fill=gray,text=white}]
\node[draw=none] (q8) [below left=0.7cm and 0.3cm of s3p] {};
\draw [black] (s3p) -- (q8);
\node[draw=none] (s8) [below right=0.7cm and 0.3cm of s3p,draw = none] {};
\draw [black] (s3p) -- (s8); 
\end{scope}
\node[draw=none] (q8m) [below left=0.7cm and 0.0cm of s3p, draw = none] {\ldots};
}
\begin{figure}[H]
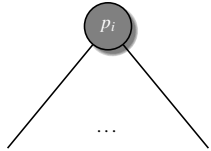

\centering
\prs
\caption{Sub-partial run rooted at $p_i$.}
\label{fig:matteo2_bis_bis}
\end{figure}
Our goal now is to derive the profile $\synseq {p_i} q {w_i}$.
As in the previous case, if there are no states labeled by $q$ in $\rho_i$, then the partial run $\rho_i$ actually has profile $\synsneq {p_i} q {w_i}$ and this is derivable by inductive hypothesis on $q$ with a proof of size $f(q-1, |w_i|)$.
Therefore let us assume that some states in $\rho_i$ are labeled by $q$.

We distinguish two cases: $w_i\!\neq\! \emptyset$ and $w_i\!=\!\emptyset$, where $w_i$ is the multiset of ports in $\rho_i$.
If $w_i \!\neq\! \emptyset$ then our argument proceeds exactly as in the previous case with the only difference that the derivation of the profile $\synsneq {p_i} {q} {w_i}$ is obtained by application of the Strong Looping rule (SL) rather than (WL). Hence (as in Equation \ref{equation_bound_3}) we have a proof of size 
\begin{equation}
M_i  \leq f(q-1, 2|w|) + (f(q-1, 2|w|)\cdot 2^{|w|} ) + |w|^2.
\label{equation_bound_3_bis}
\end{equation}

Consider then the case $w_i\!=\! \emptyset$, i.e., the case of $\rho_i$ having profile $\synsneq {p_i} {q} {\emptyset}$. Since $q\!\in\! \qzero$, the set of paths in $\rho_i$ having infinitely many occurrences of $q$ has probability $0$.

\begin{claim} Since $\rho_i$ is a \emph{regular tree}, there exists a subtree $\rho^\prime$ of $\rho_i$ that does not contain any vertex labeled by $q$. 
\end{claim}
\begin{proof}
By contradiction, if vertices labeled by $q$ are reachable by all states in the tree $\rho_i$, then by regularity (i.e., $\rho_i$ can be represented as a finite graph), the set of paths visiting infinitely many $q$'s has probability $1$.
\end{proof}

\begin{remark*}
The above claim is the only point in the proof where the regularity assumption is used.
\end{remark*}

Let $r$ be the root of the subtree. 
The subtree $\rho^\prime$ has profile $\synsneq {r} {q} {\emptyset}$ (note the strict inequality) which is derivable by induction on $q$ with a proof of size at most $f(q-1,0)$.
Now, from the partial run $\rho_i$, remove all the vertices below the state $r$. In this way we obtain a new partial run having profile $\synseq {p_i} {q} {\{r\}}$. We can derive this profile as described in the previous part of the proof (case $w_i\!\neq\! \emptyset$, see Equation \ref{equation_bound_3_bis}) with a derivation having size $ \leq f(q-1, 2|\{r\}|) + (f(q-1, 2|\{r\}|)\cdot 2^{|\{r\}|}) + |\{r\}|^2 = f(q-1,2)+ 2f(q-1,2) + 1 = 3f(q-1,2) + 1$.
Then, by application of the rule $(U$) we can obtain the desired derivation:

\begin{prooftree}
\AxiomC{$\synseq {p_i} {q} {\{r\}}$}
\AxiomC{$\synsneq {r} {q} {\emptyset}$}
\LeftLabel{Unification (U):}
\BinaryInfC{$\synseq {p_i} {q} {\emptyset}$}
\end{prooftree}
The resulting proof has size 
\begin{equation}
M_i  \leq (3f(q-1,2) + 1)+ f(q-1,0)+1.
\label{equation_bound_4}
\end{equation}
This concludes the proof regarding the derivability of $\synseq {p_i} q {w_i}$ of $\rho_i$. That is, we have established in this Subsection 
(Subsection \ref{subsection_inductive_step}, see Equations \ref{equation_bound_2},  \ref{equation_bound_3}, \ref{equation_bound_3_bis} and  \ref{equation_bound_4})  that, for each $p_i\!\in\! P$, the profile $\synseq {p_i} q {w_i}$ of $\rho_i$ can be proved by a derivation of size 

$$M_i \leq \max
\begin{cases}
f(q-1, 2|w|) + h(2^{|w|})\cdot |w| + |w|^2,\\
f(q-1, 2|w|)\cdot (2^{|w|} + 1) + |w|^2,\\
(3f(q-1,2) + 1)+ f(q-1,0)+1.
\end{cases}$$

From Equation \ref{equation_bound_1} we know that 
$$N \leq f(q-1, 2|w|) + g(2^{|w|+|Q|})\cdot {|w|} + |w|^2.$$ Therefore, following Equation \ref{main_equation_inductive}, we know that $K=\max\{N,M_i\}$ is bounded by 

\begin{equation}
 K \leq\begin{cases}
f(q-1, 2|w|) + h(2^{|w|})\cdot |w| + |w|^2,\\
f(q-1, 2|w|)\cdot (2^{|w|} + 1) + |w|^2,\\
(3f(q-1,2) + 1)+ f(q-1,0)+1,\\
N \leq f(q-1, 2|w|) + g(2^{|w|+|Q|})\cdot {|w|} + |w|^2
\end{cases}
 \label{final_inequality_last}
 \end{equation}

\noindent
This completes the proof of Theorem \ref{soundness_and_completeness}.

\bibliographystyle{eptcs}
\bibliography{bib}



\end{document}